\documentclass{article}

\usepackage{epsfig}
\usepackage{amsfonts}
\usepackage{amssymb}
\usepackage{amstext}
\usepackage{amsmath}
\usepackage{amsthm,thmtools,enumerate,float,algorithm,algorithmic}
\usepackage[usenames,dvipsnames,svgnames,table]{xcolor}

\usepackage[pagebackref,letterpaper=true,colorlinks=true,pdfpagemode=none,citecolor=OliveGreen,linkcolor=BrickRed,urlcolor=BrickRed,pdfstartview=FitH]{hyperref}

 
\newcommand{\opt}[0]{{\ensuremath{\sf{opt}}}}

\numberwithin{equation}{section}
\numberwithin{table}{section}

\newtheorem{fact}{Fact}[section]
\newtheorem{lemma}[fact]{Lemma}
\newtheorem{theorem}[fact]{Theorem}

\newtheorem{corollary}[fact]{Corollary}
\newtheorem{proposition}[fact]{Proposition}
\newtheorem{claim}[fact]{Claim}

\newcommand{\R}{\ensuremath{\mathbb R}}

\newcommand{\E}[1]{{\mathbb{E}}\left[#1\right]}

\newcommand{\junk}[1]{}

\renewcommand{\l}{\lambda}
\renewcommand{\L}{{\mathcal L}}
\newcommand{\vol}{{\rm vol}}

\providecommand{\norm}[1]{\left\lVert#1\right\rVert}

\newenvironment{proofof}[1]{{\medbreak\noindent \em Proof of #1.  }}{\hfill\qed\medbreak}

\def\b1{{\bf 1}}
\def\eps{{\epsilon}}

\def\cM{{\cal M}}
\def\cE{{\cal E}}
\def\cL{{\cal L}}
\def\R{\mathbb{R}}
\def\cR{{\mathcal{R}}}

\def\mass{\ell}

\def\sqt{{\sqrt{t}}}

\DeclareMathOperator{\sspan}{span}

\DeclareMathOperator{\len}{len}

\DeclareMathOperator{\argmin}{argmin}
\DeclareMathOperator{\dist}{dist}
\DeclareMathOperator{\range}{range}
\DeclareMathOperator{\supp}{supp}

\title{Improved Cheeger's Inequality: \\Analysis of Spectral Partitioning Algorithms \\through Higher Order Spectral Gap}

\author{ 
Tsz Chiu Kwok\thanks{The Chinese University of Hong Kong. Supported by Hong Kong RGC grant 2150701.  Email: \protect\url{tckwok@cse.cuhk.edu.hk}}
\and
Lap Chi Lau\thanks{The Chinese University of Hong Kong.  Supported by Hong Kong RGC grant 2150701.  Email: \protect\url{chi@cse.cuhk.edu.hk}}
\and
Yin Tat Lee\thanks{The Chinese University of Hong Kong.  Currently a PhD student of MIT.  Email: \protect\url{yintat@mit.edu}}
\and
Shayan Oveis Gharan\thanks{Department of Management Science and Engineering, Stanford University. Supported by a Stanford Graduate Fellowship. Email: \protect\url{shayan@stanford.edu}}
\and
Luca Trevisan\thanks{Department of Computer Science, Stanford University. This material is based upon  work supported by the National Science Foundation under grant No.  CCF 1017403. Email: \protect\url{trevisan@stanford.edu}}
}

\date{}

\begin{document}

\maketitle

\begin{abstract}
Let $\phi(G)$ be the minimum conductance of an undirected graph $G$, and let $0=\lambda_1 \leq \lambda_2 \leq \ldots \leq \lambda_n \leq 2$ be the eigenvalues of the normalized Laplacian matrix of $G$. 
We prove that for any graph $G$ and any $k\geq 2$, 
\[\phi(G) = O(k) \frac{\l_2}{\sqrt{\l_k}},\] 
and this performance guarantee is achieved by the spectral partitioning algorithm.
This improves Cheeger's inequality, 
and the bound is optimal up to a constant factor for any $k$.
Our result shows that the spectral partitioning algorithm is a constant factor approximation algorithm for finding a sparse cut if $\l_k$ is a constant for some constant $k$.
This provides some theoretical justification to its empirical performance in image segmentation and clustering problems.
We extend the analysis to other graph partitioning problems, including multi-way partition, balanced separator, and maximum cut.
\end{abstract}

\newpage

\tableofcontents

\newpage

\section{Introduction}

We study the performance of spectral algorithms for graph partitioning problems.  
For the moment, we assume the graphs are unweighted and $d$-regular for simplicity, while the results in the paper hold for arbitrary weighted graphs, with suitable changes to the definitions.
Let $G=(V,E)$ be a $d$-regular undirected graph. 
The {\em conductance} of a subset $S \subseteq V$ is defined as 
\[\phi(S) = \frac{|E(S,\overline{S})|}{d \min\{|S|,|\overline{S}|\}},\]
where $E(S,\overline{S})$ denotes the set of edges of $G$ crossing from $S$ to its complement.
The conductance of the graph $G$ is defined as 
\[\phi(G) = \min_{S\subset V} \phi(S).\]
Finding a set of small conductance, also called a sparse cut, is an algorithmic problem that comes up in different areas of computer science.
Some applications include image segmentation~\cite{shi-malik,tolliver-miller}, clustering~\cite{ng-jordan-weiss,kannan-vempala-vetta,von-luxburg}, community detection~\cite{leskovec-lang-mahoney}, and designing approximation algorithms~\cite{shmoys}.

A fundamental result in spectral graph theory provides a connection between the conductance of a graph and the second eigenvalue of its normalized Laplacian matrix.
The normalized Laplacian matrix $\cL \in \mathbb R^{V \times V}$ is defined
as $\cL = I - \frac{1}{d} A$, where $A$ is the adjacency matrix of $G$.  
The eigenvalues of $\cL$ satisfy $0 = \l_1 \leq \l_2 \leq \ldots \leq \l_{|V|} \leq 2$.
It is a basic fact that $\phi(G)=0$ if and only if $\lambda_2=0$.
Cheeger's inequality for graphs 
provides a quantitative generalization of this fact:
\begin{equation}
\label{eq:Cheegerineq}
\frac12 \l_2 \leq \phi(G) \leq \sqrt{2\l_2}.
\end{equation}
This is first proved in the manifold setting by Cheeger~\cite{cheeger} and is extended to undirected graphs by Alon and Milman~\cite{alon-milman,alon}.
Cheeger's inequality is an influential result in spectral graph theory with applications in spectral clustering \cite{spielman-teng,kannan-vempala-vetta}, explicit construction of expander graphs \cite{jimbo-maruoka,hoory-linial-wigderson,lee12}, approximate counting \cite{sinclair-jerrum,jerrum-sinclair-vigoda}, and image segmentation \cite{shi-malik}.

We improve Cheeger's inequality using higher eigenvalues of the normalized Laplacian matrix.

\begin{theorem} \label{t:maincheeger}
For every undirected graph $G$ and any $k \geq 2$, it holds that 
\[\phi(G) = O(k) \frac{\l_2}{\sqrt{\l_k}}.\]
\end{theorem}

This improves Cheeger's inequality, as
it shows that $\l_2$ is a better approximation of $\phi(G)$ when there is a large gap between $\l_2$ and $\l_k$ for any $k \geq 3$.
The bound is optimal up to a constant factor for any $k \geq 2$, 
as the cycle example shows that $\phi(G) = \Omega (k \l_2/\sqrt{\l_k})$ for any $k \geq 2$.

\subsection{The Spectral Partitioning Algorithm}

The proof of Cheeger's inequality is constructive and it gives the following simple nearly-linear time algorithm (the {\em spectral partitioning algorithm}) that finds cuts with approximately minimal conductance.
Compute the second eigenfunction $g \in {\mathbb R}^V$ of the normalized Laplacian matrix $\cL$, and let $f = g/\sqrt{d}$. 
For a threshold $t\in \mathbb{R}$, let $V(t):=\{v: f(v) \geq t\}$ be a threshold set of $f$.
Return the threshold set of $f$ with the minimum conductance among all thresholds $t$.
Let $\phi(f)$ denote the conductance of the return set of the algorithm.
The proof of Cheeger's inequality shows that $\frac{1}{2} \l_2 \leq \phi(f) \leq \sqrt{2\l_2}$, and hence
the spectral partitioning algorithm is a nearly-linear  time $O(1/\sqrt{\l_2})$-approximation algorithm for finding a sparse cut.
In particular, it gives a constant factor approximation algorithm when $\l_2$ is a constant, but since $\l_2$ could be as small as $1/n^2$ even for a simple unweighted graph (e.g. for the cycle), the performance guarantee could be $\Omega(n)$.

We prove \autoref{t:maincheeger} by showing a stronger statement, that is $\phi(f)$ is upper-bounded by $O(k\lambda_2/\sqrt{\lambda_k})$.
\begin{theorem}
\label{t:main}
For any undirected graph $G$, and $k\geq 2$, 
$$ \phi(f) = O(k)\frac{\lambda_2}{\sqrt{\lambda_k}}.$$
\end{theorem}
This shows that the spectral partitioning algorithm is a $O(k/\sqrt{\l_k})$-approximation algorithm for the sparsest cut problem, even though it does not employ any information about higher eigenvalues.
In particular, spectral partitioning provides a constant factor approximation for the sparsest cut problem when $\l_k$ is a constant for some constant $k$.

\subsection{Generalizations of Cheeger's Inequality}

There are several recent results showing new connections between the expansion profile of a graph and the higher eigenvalues of its normalized Laplacian matrix.
The first result in this direction is about the small set expansion problem. Arora, Barak and Steurer~\cite{arora-barak-steurer} show that
if there are $k$ small eigenvalues for some large $k$, then the graph has a sparse cut $S$ with $|S| \approx n/k$.
In particular, if $k = |V|^{\epsilon}$ for $\epsilon \in (0,1)$, then the graph has a sparse cut $S$ with $\phi(S) \leq O(\sqrt{\l_k})$ and $|S| \approx n/k$.
This can be seen as a generalization of Cheeger's inequality to the small set expansion problem (see~\cite{steurer,gharan-trevisan,odonnell-witmer} for some improvements).

Cheeger's inequality for graph partitioning can also be extended to higher-order Cheeger's inequality for $k$-way graph partitioning~\cite{louis+2,lee-gharan-trevisan}: 
If there are $k$ small eigenvalues, then there are $k$ disjoint sparse cuts.
Let 
$$\phi_k(G) := \min_{S_1, \ldots, S_k} \max_{1\leq i\leq k}\phi(S_i)$$ 
where $S_1, \ldots, S_k$ are over non-empty disjoint subsets $S_1, \ldots, S_k \subseteq V$.
Then
\[\frac{1}{2} \l_k \leq \phi_k(G) \leq O(k^2) \sqrt{\l_k}.\]

Our result can be applied to $k$-way graph partitioning by combining with a result in~\cite{lee-gharan-trevisan}.

\begin{corollary} \label{c:multiway}
For every undirected graph $G$ and any $l > k \geq 2$, it holds that
\begin{enumerate}[(i)]
\item \[\phi_{k}(G) \leq O(l k^6) \frac{\l_k}{\sqrt{\l_l}}.\]
\item For any $\delta \in (0,1)$, 
$$ \phi_{(1-\delta)k}(G) \leq O\left( \frac{l\log^2{k}}{\delta^8 k}\right) \frac{\lambda_k}{\sqrt{\lambda_l}}.$$
\item If $G$ excludes $K_h$ as a minor, then for any $\delta\in (0,1)$ 
$$ \phi_{(1-\delta)k}(G) \leq O\left(\frac{h^4l}{\delta^5k}\right)  \frac{\lambda_k}{\sqrt{\lambda_l}}.$$
\end{enumerate}
\end{corollary}

Part (i) shows that $\l_k$ is a better approximation of $\phi_k(G)$ when there is a large gap between $\l_k$ and $\l_l$ for any $l > k$.
Part (ii) implies that $\phi_{0.9k}(G) \leq O(\l_k \log^2k/\sqrt{\l_{2k}})$, and similarly part (iii) implies that $\phi_{0.9k}(G) \leq O(\l_k/\sqrt{\l_{2k}})$ for planar graphs.

Furthermore, our proof shows that the spectral algorithms in~\cite{lee-gharan-trevisan} achieve the corresponding approximation factors. 
For instance, when $\l_l$ is a constant for a constant $l > k$, there is  a constant factor approximation algorithm for the $k$-way partitioning problem.

\subsection{Analysis of Practical Instances}

Spectral partitioning is a popular heuristic in practice, as it is easy to be implemented and can be solved efficiently by standard linear algebra methods.
Also, it has good empirical performance in applications including image segmentation~\cite{shi-malik} and clustering~\cite{von-luxburg}, much better than the worst case performance guarantee provided by Cheeger's inequality.
It has been an open problem to explain this phenomenon rigorously~\cite{spielman-teng,guattery-miller}.
There are some research directions towards this objective.

One direction is to analyze the average case performance of spectral partitioning.
A well-studied model is the random planted model~\cite{boppana,alon-krivelevich-sudakov,mcsherry},
where there is a hidden bisection $(X,Y)$ of $V$ and there is an edge between two vertices in $X$ and two vertices in $Y$ with probability $p$ and there is an edge between a vertex in $X$ and a vertex in $Y$ with probability $q$.
It is proved that spectral techniques can be used to recover the hidden partition with high probability, as long as $p-q \geq \Omega(\sqrt{p \log |V|/|V|})$~\cite{boppana,mcsherry}.
The spectral approach can also be used for other hidden graph partitioning problems~\cite{alon-krivelevich-sudakov,mcsherry}.
Note that the spectral algorithms used are usually not exactly the same as the spectral partitioning algorithm.
Some of these proofs explicitly or implicitly use the fact that there is a gap between the second and the third eigenvalues. See \autoref{s:semirandom} for more details.

To better model practical instances, 
Bilu and Linial~\cite{bilu-linial} introduced the notion of stable instances for clustering problems.
One definition for the sparsest cut problem is as follows: an instance is said to be $\gamma$-stable if there is an optimal sparse cut $S \subseteq V$ which will remain optimal even if the weight of each edge is perturbed by a factor of $\gamma$.
Intuitively this notion is to capture the instances with an outstanding solution that is stable under noise, and arguably they are the meaningful instances in practice.
Note that a planted bisection instance is stable if $p-q$ is large enough, and so this is a more general model than the planted random model.
Several clustering problems are shown to be easier on stable instances~\cite{balcan-blum-gupta,awasthi-blum-sheffet,daniely-linial-saks}, and spectral techniques have been analyzed for the stable maximum cut problem~\cite{bilu-linial,bilu-daniely-linial-saks}.
See \autoref{s:stable} for more details.

Informally, the higher order Cheeger's inequality shows that an undirected graph has $k$ disjoint sparse cuts if and only if $\l_k$ is small. 
This suggests that the graph has at most $k-1$ outstanding sparse cuts when $\l_{k-1}$ is small and $\l_k$ is large.
The algebraic condition that $\l_2$ is small and $\l_3$ is large seems similar to the stability condition but more adaptable to spectral analysis.
This motivates us to analyze the performance of the spectral partitioning algorithm through higher-order spectral gaps.

In practical instances of image segmentation, there are usually only a few outstanding objects in the image, and so $\l_k$ is large for a small $k$~\cite{von-luxburg}.
Thus \autoref{t:main} provides a theoretical explanation to why the spectral partitioning algorithm performs much better than the worst case bound by Cheeger's inequality in those instances.
In clustering applications, there is a well-known eigengap heuristic that partitions the data into $k$ clusters if $\l_k$ is small and $\l_{k+1}$ is large~\cite{von-luxburg}.
\autoref{c:multiway} shows that in such situations the spectral algorithms in~\cite{lee-gharan-trevisan} perform better than the worst case bound by the higher order Cheeger's inequality.

\subsection{Other Graph Partitioning Problems}

Our techniques can be used to improve the spectral algorithms for other graph partitioning problems using higher order eigenvalues.
In the minimum bisection problem, the objective is to find a set $S$ with minimum conductance among the sets with $|V|/2$ vertices.
While it is very nontrivial to find a sparse cut with exactly $|V|/2$ vertices~\cite{feige-krauthgamer, racke},
it is well known that a simple recursive spectral algorithm can find a {\em balanced separator} $S$ with $\phi(S) = O(\sqrt{\eps})$ with $|S| = \Omega(|V|)$, where $\eps$ denotes the conductance of the minimum bisection (e.g. ~\cite{kannan-vempala-vetta}).
We use \autoref{t:main} to generalize the recursive spectral algorithm to obtain a better approximation guarantee when $\l_k$ is large for a small $k$.

\begin{theorem}
\label{thm:minbisection-intro}
Let 
$$\eps:=\min_{|S|=|V|/2} \phi(S).$$
There is a polynomial time algorithm
that finds a set $S$ such that $|V|/5 \leq |S| \leq 4|V|/5$ and $\phi(S)\leq O(k \eps/\lambda_k)$.
\end{theorem}

In the maximum cut problem, the objective is to find a partition of the vertices which maximizes the weight of edges whose endpoints are on different sides of the partition.
Goemans and Williamson~\cite{goemans-williamson} gave an SDP-based $0.878$-approximation algorithm for the maximum cut problem.
Trevisan~\cite{trevisan09} gave a spectral algorithm with approximation ratio strictly better than $1/2$.
Both algorithms find a solution that cuts at least $1-O(\sqrt{\eps})$ fraction of edges when the optimal solution cuts at least $1-O(\eps)$ fraction of edges.
Using a similar method as in the proof of \autoref{t:main},
we generalize the spectral algorithm in~\cite{trevisan09} for the maximum cut problem to obtain a better approximation guarantee when $\l_{n-k}$ is small for a small $k$.

\begin{theorem}
\label{thm:maxcut-intro}
There is a polynomial time algorithm that on input graph $G$ finds a cut $(S,\overline{S})$ such that if the optimal solution cuts at least $1-\eps$ fraction of the edges, 
then $(S,\overline{S})$ cuts at least 
$$1-O(k)\log\Big(\frac{2-\lambda_{n-k}}{k\eps}\Big)\frac{\eps}{2-\lambda_{n-k}}$$ 
fraction of edges. 
\end{theorem}

\subsection{More Related Work} \label{s:related}

{\em Approximating Graph Partitioning Problems:}
Besides spectral partitioning, there are approximation algorithms for the sparsest cut problem based on linear and semidefinite programming relaxations.
There is an LP-based $O(\log n)$ approximation algorithm by Leighton and Rao~\cite{leighton-rao}, and an SDP-based $O(\sqrt{\log n})$ approximation algorithm by Arora, Rao and Vazirani~\cite{arora-rao-vazirani}.
The subspace enumeration algorithm by Arora, Barak and Steurer~\cite{arora-barak-steurer} provides an $O(1/\lambda_k)$ approximation algorithm for the sparsest cut problem with running time $n^{O(k)}$,  by searching for a sparse cut in the $(k-1)$-dimensional eigenspace corresponding to $\lambda_1,\ldots,\lambda_{k-1}$.
It is worth noting that for $k=3$ the subspace enumeration algorithm is exactly the same as the spectral partitioning algorithm. 
Nonetheless, the result in~\cite{arora-barak-steurer} is incomparable to \autoref{t:main} since it does not upper-bound $\phi(G)$ by a function of $\l_2$ and $\l_3$.
Recently, using the Lasserre hierarchy for SDP relaxations, Guruswami and Sinop~\cite{guruswami-sinop2} gave an $O(1/\l_k)$ approximation algorithm for the sparsest cut problem with running time $n^{O(1)}2^{O(k)}$. 
Moreover, the general framework of Guruswami and Sinop~\cite{guruswami-sinop2} applies to other graph partitioning problems including minimum bisection and maximum cut, obtaining approximation algorithms with similar performance guarantees and running times.
This line of recent work is closely related to ours in the sense that it shows that many graph partitioning problems are easier to approximate on graphs with fast growing spectrums, i.e. $\l_k$ is large for a small $k$.
Although their results give much better  approximation guarantees when $k$ is large,  our results show that simple spectral algorithms provide nontrivial performance guarantees.

{\em Higher Eigenvalues of Special Graphs:}
Another direction to show that spectral algorithms work well is to analyze their performance in special graph classes.
Spielman and Teng~\cite{spielman-teng} showed that $\l_2 = O(1/n)$ for a bounded degree planar graph and a spectral algorithm can find a separator of size $O(\sqrt{n})$ in such graphs.
This result is extended to bounded genus graphs by Kelner~\cite{kelner} and to fixed minor free graphs by Biswal, Lee and Rao~\cite{biswal-lee-rao}.
This is further extended to higher eigenvalues by Kelner, Lee, Price and Teng~\cite{kelner-lee-price-teng}: $\l_k = O(k/n)$ for planar graphs, bounded genus graphs, and fixed minor free graphs when the maximum degree is bounded.
Combining with a higher order Cheeger inequality for planar graphs~\cite{lee-gharan-trevisan}, this implies that $\phi_k(G) = O(\sqrt{k/n})$ for bounded degree planar graphs.
We note that these results give mathematical bounds on the conductances of the resulting partitions, but they do not imply that the approximation guarantee of Cheeger's inequality could be improved for these graphs,
neither does our result as these graphs have slowly growing spectrums.

{\em Planted Random Instances, Semi-Random Instances, and Stable Instances:}
We have discussed some previous work on these topics, and we will discuss some relations to our results in \autoref{s:semirandom} and \autoref{s:stable}.

\subsection{Proof Overview}

We start by describing an  informal intuition of the proof of \autoref{t:main} for $k=3$, and 
then we describe how this intuition can be generalized. For a function  $f\in \mathbb{R}^V$,
let $\cR(f)=f^TLf/(d\norm{f}^2)$ be the Rayleigh quotient of $f$ (see (\ref{eq:rayleighquotient}) of \autoref{s:spectral} for the definition in general graphs). Let $f$ be a function that is orthogonal to the constant function and that  $\cR(f)\approx \lambda_2$.

Suppose $\l_2$ is small and $\l_3$ is large.
Then the higher order Cheeger's inequality implies that  there is a partitioning of the graph into two sets of small conductance, but in every partitioning into at least three sets, there is a set of large conductance.
So, we expect the graph to have a sparse cut of which the two parts are expanders; see~\cite{tanaka} for a quantitative statement.
Since $\cR(f)$ is small and  $f$ is orthogonal to the constant function, 
we expect that the vertices in the same expander have similar values in $f$ and the average values of the two expanders are far apart.
Hence, $f$ is similar to a step function with two steps representing a cut, and we expect that $\cR(f) \approx \phi(G)$ in this case.
Therefore, roughly speaking, $\l_3 \gg \l_2$ implies  $\l_2 \approx \phi(G)$.

Conversely, \autoref{t:main} shows that if $\l_2 \approx \phi^2(G)$ then $\l_3 \approx \l_2$.
One way to prove that $\l_2 \approx \l_3$ is to find a function $f'$  of Rayleigh quotient close to $\lambda_2$ such that $f'$ is  orthogonal to both $f$ and the constant function.
For example, if $G$ is a cycle, then  $\l_2 = \Theta(1/n^2)$, $\phi(G) = \Theta(1/n)$,
and $f$ (up to normalizing factors) could represent the cosine function. In this case we may define $f'$ to be the sine function. 
Unfortunately, finding such a function $f'$ in general is not as straightforward. 
Instead, our idea is to find three {\em disjointly supported} functions $f_1,f_2,f_3$ of Rayleigh quotient close to $\lambda_2$. 
As we 
prove in \autoref{c:orthogonal}, 
this would upper-bound $\lambda_3$ by $2\max\{\cR(f_1),\cR(f_2),\cR(f_3)\}$. 
For the cycle example, if $f$ is the cosine function, we may construct $f_1,f_2,f_3$ simply
by first dividing the support of $f$ into three disjoint intervals and then constructing each $f_i$ by
defining a smooth localization of $f$ in one of those intervals. 
To ensure that $\max\{\cR(f_1),\cR(f_2),\cR(f_3)\}\approx \lambda_2$ we need to show that
$f$ is a ``smooth'' function, whose values change continuously. 
We make this rigorous by showing that  if $\l_2 \approx \phi(G)^2$, then the function $f$ must be smooth.
Therefore, we can construct three disjointly supported functions based on $f$ and show that $\l_2 \approx \l_3$.

We provide two proofs of \autoref{t:main}. The first proof generalizes the first observation. We show that if $\lambda_k \gg k\lambda_2$, then $\phi(G)\approx k\lambda_2$.
The main idea is to show that if $\lambda_k \gg k\lambda_2$, then $f$ can be approximated
by a $k$ step function $g$ in the sense that $\norm{f-g} \approx 0$ (in general we show that any function $f$ can be approximated by a $k$ step function $g$  such that any $\norm{f-g}^2 \leq \cR(f)/\lambda_k$).
It is instructive to prove that if $f$ is exactly a $k$-step function then $\phi(G) \leq O(k\cR(f))$. 
Our main technical step, \autoref{l:jump}, provides a robust version of the latter fact by showing that for any $k$-step approximation of $f$, $\phi(f) \leq O(k (\cR(f) + \norm{f-g} \sqrt{\cR(f)}))$.

On the other hand, our second proof generalizes the second observation. Say $\cR(f) \approx \phi(G)^2$. We partition the support of $f$ into disjoint intervals of the form $[2^{-i},2^{-(i+1)}]$,
and we show that the vertices are distributed almost uniformly in most of these intervals
in the sense that if we divide $[2^{-i},2^{-(i+1)}]$ into $k$ equal length subintervals, then we expect to see the same amount of mass in the subintervals. 
This shows that $f$ is a smooth function. 
We then argue that $\lambda_k \lesssim k\lambda_2$, by constructing $k$ disjointly supported
functions each of Rayleigh quotient $O(k^2)\cR(f)$.

\section{Preliminaries} \label{s:prelim}
Let $G=(V,E,w)$ be a finite, undirected graph, with positive weights $w : E \to (0,\infty)$ on the edges.
For a pair of vertices $u,v \in V$, we write $w(u,v)$ for $w(\{u,v\})$.
For a subset of vertices $S \subseteq V$, we write $E(S):=\{ \{u,v\}\in E: u,v\in S\}$. 
For disjoint sets $S,T\subseteq V$, we write $E(S,T) := \{ \{u,v\} \in E : u\in S,v\in T \}$.
For a subset of edges $F \subseteq E$, we write $w(F) = \sum_{e \in F} w(e)$.
We use $u \sim v$ to denote $\{u,v\} \in E$.
We extend the weight to vertices by defining, for a single vertex $v\in V$, $w(v):=\sum_{u\sim v} w(u,v)$. We can think of $w(v)$ as the weighted degree of vertex $v$.
For the sake of clarity we will assume throughout that $w(v) \geq 1$ for every $v \in V$.  For $S \subseteq V$,
we write $\vol(S) = \sum_{v \in S} w(v)$ to denote the {\em volume} of $S$.

Given a subset $S \subseteq V$, 
we denote the {\em Dirichlet conductance} of $S$ by
$$
\phi(S) := \frac{w(E(S,\overline{S}))}{\min\{\vol(S),\vol(\overline{S})\}}\,.
$$
For a function $f \in \mathbb{R}^V$, and a threshold $t\in \mathbb{R}$, let $V_f(t):=\{v: f(v) \geq t\}$ be a threshold set of $f$.
We let 
$$\phi(f) := \min_{t\in\mathbb{R}} \phi(V_f(t)).$$
be the conductance of the best threshold set of the function $f$, and $V_f(t_\opt)$ be the smaller side
(in volume) of that minimum cut.

For any two thresholds $t_1,t_2 \in {\mathbb R}$, we use
\[[t_1,t_2]:=\{x\in \mathbb{R}: \min\{t_1,t_2\} < x \leq \max\{t_1,t_2\} \}.\]
Note that all intervals are defined to be closed on the larger value and open on the smaller value. 
For an interval $I=[t_1,t_2]\subseteq \mathbb{R}$, we use $\len(I):=|t_1-t_2|$ to denote the length of $I$.
For a function $f \in {\mathbb R}^V$,
we define  $V_f(I):=\{v: f(v) \in I\}$ to denote the vertices within $I$. The volume of an  interval $I$ is defined as $\vol_f(I):=\vol(V_f(I))$. We also abuse the notation and use
 $\vol_f(t) := \vol(V_f(t))$ to denote the volume of the interval $[t,\infty]$. 
We define the {\em support} of $f$, $\supp(f) := \{v:f(v) \neq 0\}$, as the set of vertices with nonzero values in $f$. 
We say two functions $f,g\in \mathbb{R}^V$ are disjointly supported if $\supp(f)\cap \supp(g)=\emptyset$. 

For any $t_1, t_2,\ldots,t_l\in \mathbb{R}$, let $\psi:\mathbb{R}\rightarrow \mathbb{R}$ be defined as
$$\psi_{t_1,\ldots,t_l}(x) = \argmin_{t_i} |x-t_i|.$$
In words, for any $x\in \R$, $\psi_{t_1,\ldots,t_i}(x)$ is the value of $t_i$ closest to $x$.

For $\rho>0$, we say a function $g$ is {\em$\rho$-Lipschitz} w.r.t. $f$, if for all pairs of vertices $u,v\in V$,
$$ |g(u) - g(v)|\leq \rho |f(u)-f(v)|.$$ 

The next inequality follows from the Cauchy-Schwarz inequality and will be useful in our proof.
Let $a_1,\ldots,a_m, b_1,\ldots,b_m \geq 0$. Then,
\begin{equation}
\label{eq:cauchy}
\sum_{i=1}^m \frac{a_i^2}{b_i} \geq \frac{\left(\sum_{i=1}^m a_i\right)^2}{\sum_{i=1}^m b_i}.
\end{equation}

\subsection{Spectral Theory of the Weighted Laplacian} \label{s:spectral}

We write $\ell^2(V, w)$ for the Hilbert space of functions $f : V \to \mathbb R$ with
inner product $$\langle f,g \rangle_{w} := \sum_{v \in V} w(v) f(v) g(v),$$
and norm $\|f\|_{w}^2 = \langle f,f\rangle_{w}$.
We reserve $\langle \cdot , \cdot \rangle$ and $\|\cdot\|$ for the standard
inner product and norm on $\mathbb R^k$, $k \in \mathbb N$ and $\ell^2(V)$.

We consider some operators on $\ell^2(V,w)$.  The adjacency operator
is defined by $A f(v) = \sum_{u \sim v} w(u,v) f(u)$, and the diagonal degree operator by
$D f(v) = w(v) f(v)$.  Then the {\em combinatorial Laplacian} is defined by
$L = D-A$, and the {\em normalized Laplacian} is given by
$$\mathcal L_G := I - D^{-1/2} A D^{-1/2}.$$
Observe that for a $d$-regular unweighted graph, we have $\mathcal L_G = \frac{1}{d} L$.

If $g : V \to \mathbb R$ is a non-zero function and $f = D^{-1/2} g$, then
\begin{eqnarray}
\frac{\langle g, \mathcal L_G \,g\rangle}{\langle g,g\rangle}
= \frac{\langle g, D^{-1/2} L D^{-1/2} g\rangle}{\langle g,g\rangle} 
= \frac{\langle f, L f\rangle}{\langle D^{1/2} f, D^{1/2} f\rangle} 
= \frac{\displaystyle \sum_{u \sim v} w(u,v) |f(u)-f(v)|^2}{\displaystyle \norm{f}_w^2}
=: \cR_G(f) 
\label{eq:rayleighquotient}
\end{eqnarray}
where the latter value is referred to as the {\em Rayleigh quotient of $f$ (with respect to $G$)}.
We drop the subscript of $\cR_G(f)$ when the graph is clear in the context. 

In particular, $\mathcal L_G$ is a positive-definite operator with
eigenvalues $$0 = \lambda_1 \leq \lambda_2 \leq \cdots \leq \lambda_n \leq 2\,.$$
For a connected graph, the first eigenvalue corresponds
to the eigenfunction $g = D^{1/2} f$, where $f$ is any non-zero constant function.
Furthermore, by standard variational principles,
\begin{eqnarray}
\lambda_k &=& \min_{g_1, \ldots, g_k \in \ell^2(V)} \max_{g \neq 0} \left\{ \frac{\langle g, \mathcal L_G\, g\rangle}{\langle g, g\rangle} : g \in \mathrm{span}\{g_1, \ldots, g_k\}\right\} \nonumber \\
&=& \min_{f_1, \ldots, f_k \in \ell^2(V,w)} \max_{f \neq 0}  \left\{ \vphantom{\bigoplus} \cR(f) : f \in \mathrm{span}\{f_1, \ldots, f_k\}\right\}, \label{eq:eigenvar}
\end{eqnarray}
where both minimums are over sets of $k$ non-zero orthogonal functions in the Hilbert spaces $\ell^2(V)$ and $\ell^2(V,w)$, respectively.
We refer to \cite{Chung97} for more background on the spectral theory
of the normalized Laplacian.
The following proposition is proved in~\cite{hoory-linial-wigderson} and will be useful in our proof

\begin{proposition}[Horry, Linial and Widgerson~\cite{hoory-linial-wigderson}] \label{c:two}
There are two disjointly supported functions $f_+, f_-\in \ell^2(V,w)$ such that $f_+ \geq 0$ and $f_- \leq 0$ and $\cR(f_+) \leq \l_2$ and $\cR(f_-) \leq \l_2$.
\end{proposition}
\begin{proof}
Let $g \in \ell^2(V)$ be the second eigenfunction of $\L$.
Let $g_+ \in \ell^2(V)$ be the function with $g_+(u) = \max\{g(u),0\}$ and $g_- \in \ell^2(V)$ be the function with $g_-(u) = \min\{g(u),0\}$.
Then, for any vertex $u \in \supp(g_+)$, 
\[(\L g_+)(u) = g_+(u) - \sum_{v:v\sim u} \frac{w(u,v)g_+(v)}{\sqrt{w(u) w(v)}} \leq g(u) - \sum_{v:v\sim u} \frac{w(u,v)g(v)}{\sqrt{w(u) w(v)}} = (\L g)(u) = \l_2 \cdot g(u).\]
Therefore, 
\[
\langle g_+,\L g_+\rangle = \sum_{u \in \supp(g_+)} g_+(u) \cdot (\L g_+)(u) \leq \sum_{u \in \supp(g_+)} \l_2 \cdot g_+(u)^2 = \l_2 \cdot \norm{g_+}^2.\]
Letting $f_+ = D^{-1/2} g_+$, we get
\[\l_2 \geq \frac{\langle g_+, \L g_+\rangle}{\norm{g_+}^2}
= \frac{\langle f_+, L f_+\rangle }{\norm{f_+}_w^2} = \cR(f_+).\]
Similarly, we can define $f_- = D^{-1/2}g_-$, and show that $\cR(f_-) \leq \l_2$.
\end{proof}

By choosing either of $f_+$ or $f_-$ that has a smaller (in volume) support, and taking
a proper normalization, we get the following corollary. 
\begin{corollary}
\label{cor:smallsuppfunction}
There exists a function $f\in \ell^2(V,w)$ such that $f\geq 0$, $\cR(f)\leq \lambda_2$, 
$\supp(f)\leq \vol(V)/2$, and $\norm{f}_w=1$.
\end{corollary}

Instead of directly  upper bounding $\l_k$ in the proof of \autoref{t:main}, we will construct $k$ disjointly supported functions with small Rayleigh quotients.
In the next lemma we show that by the variational principle this gives an upper-bound on $\lambda_k$. 

\begin{lemma}
\label{c:orthogonal}
For any $k$ disjointly supported functions $f_1,f_2,\ldots,f_k\in\ell^2(V,w)$, we have
$$ \lambda_k \leq 2\max_{1\leq i\leq k} \cR(f_i).$$
\end{lemma}
\begin{proof}
By equation \eqref{eq:eigenvar}, it is sufficient to show that for any function $h\in \sspan\{f_1,\ldots,f_k\}$, $\cR(h)\leq \max_i \cR(f_i)$. 
Note that $\cR(f_i) = \cR(cf_i)$ for any constant $c$,
so we can assume $h:=\sum_{i=1}^k f_i$. 
Since $f_1,\ldots,f_k$ are disjointly supported, for any $u,v\in V$, we have
$$ |h(u)-h(v)|^2 \leq  \sum_{i=1}^k 2|f_i(u) - f_i(v)|^2.$$
Therefore,
\begin{eqnarray*} \cR(h) = \frac{\sum_{u\sim v} w(u,v)|h(u) - h(v)|^2}{\norm{h}_w^2}& \leq& \frac{2\sum_{u\sim v} \sum_{i=1}^k w(u,v)|f_i(u) - f_i(v)|^2}{\norm{h}_w^2} \\
&=& \frac{2\sum_{i=1}^k\sum_{u\sim v} w(u,v)|f_i(u)-f_i(v)|^2 }{ \sum_{i=1}^k \norm{f_i}_w^2}
 \leq 2\max_{1\leq i\leq k} \cR(f_i).
\end{eqnarray*}
\end{proof}


\subsection{Cheeger's Inequality with Dirichlet Boundary Conditions}

Many variants of the following lemma are known; see, e.g. \cite{Chung96}.

\begin{lemma}
\label{lem:dirichletcheeger}
For every  non-negative $h\in \ell^2(V,w)$ such that $\supp(h)\leq \vol(V)/2$, the following holds
$$\phi(h) \leq \frac {\sum_{u\sim v} w(u,v)|h(v) - h(u)| }{\sum_v w(v)h(v) }.$$
\end{lemma}
\begin{proof}
Since the right hand side is homogeneous in $h$, we may assume that $\max_v h(v)\leq 1$.
Let $0 < t\leq 1$ be chosen uniformly at random. 
Then, by linearity of expectation,
$$ \frac{\E{w(E(V_h(t),\overline{V_h(t)}))}}{\E{\vol(V_h(t))}} = \frac{ \sum_{u\sim v} w(u,v) |h(u) - h(v)| }{\sum_v w(v) h(v)}.$$
This implies that there exists a $0< t\leq 1$ such that $\phi(V_h(t))\leq \frac {\sum_{u\sim v} w(u,v) |h(v) - h(u)| }{\sum_v w(v) h(v) }$. 
The latter holds since for any $t>0$, $\vol(V_h(t)) \leq \vol(V)/2$.
\end{proof}

\subsection{Energy Lower Bound}

We define the {\em energy} of a function $f\in \ell^2(V,w)$ as
\[\cE_f := \sum_{u\sim v} w(u,v) |f(u) - f(v)|^2.\]
Observe that $\cR(f)=\cE_f/\norm{f}_w^2$.
We also define the energy of $f$ {\em restricted} to an interval $I$ as follows:
\[\cE_f(I) := \sum_{u\sim v} w(u,v)\len(I \cap [f(u),f(v)])^2.\]
When the function $f$ is clear from the context we drop the subscripts from the above definitions.

The next fact shows that  by restricting the energy of $f$ to disjoint intervals we may only decrease the energy.
\begin{fact}
\label{f:additivity}
For any set of disjoint intervals $I_1, \ldots, I_m$, we have
 $$ \cE_f \geq \sum_{i=1}^m \cE_f(I_i).$$
\end{fact}
\begin{proof}
\[\cE_f = \sum_{u\sim v} w(u,v)|f(u) - f(v)|^2 \geq \sum_{u\sim v} \sum_{i=1}^m w(u,v)\len(I_i \cap [f(u),f(v)])^2 = \sum_{i=1}^m \cE_f(I_i).
\]
\end{proof}

The following is the key lemma to lower bound the energy of a function  $f$.
It shows that a long interval with small volume must have a significant contribution
to the energy of $f$. 

\begin{lemma} \label{l:drop}
For any non-negative function $f\in \ell^2(V,w)$ with $\vol(\supp(f)) \leq \vol(V)/2$,
for any interval $I=[a,b]$ with $a > b \geq 0$, we have
\[\cE(I) \geq \frac{\phi^2(f) \cdot \vol_f^2(a) \cdot \len^2(I)}{\phi(f) \cdot \vol_f(a) + \vol_f(I)}.\]
\end{lemma}
\begin{proof}
Since $f$ is non-negative with $\vol(\supp(f)) \leq \vol(V)/2$,
by the definition of $\phi(f)$,
the total weight of the edges going out the threshold set $V_f(t)$ is at least  $\phi(f) \cdot \vol_f(a)$, for any $a \geq t \geq b \geq 0$.
Therefore, by summing over these threshold sets, we have
\[\sum_{u\sim v} w(u,v)\len(I \cap [f(u),f(v)]) \geq \len(I) \cdot \phi(f) \cdot \vol_f(a).\]
Let $E':=\{\{u,v\}: \len(I \cap [f(u),f(v)])>0\}$ be the set of edges with nonempty intersection with the interval $I$.  
Let $\beta \in (0,1)$ be a parameter to be fixed later.
Let $F\subseteq E'$ be the set of edges of $E'$ that are not adjacent to any of the vertices in $I$. If $w(F)\geq \beta w(E')$, then
\[\cE(I) \geq w(F) \cdot \len(I)^2 \geq \beta \cdot w(E') \cdot \len(I)^2 \geq \beta \cdot \phi(f) \cdot \vol_f (a) \cdot \len(I)^2.\]
Otherwise, $\vol_f(I) \geq (1-\beta)w(E')$. 
Therefore, by the Cauchy Schwarz inequality (\ref{eq:cauchy}), we have
\begin{eqnarray*}
\cE(I) = \sum_{\{u,v\} \in E'} w(u,v)(\len(I \cap [f(u),f(v)]))^2 
&\geq& \frac{\big(\sum_{\{u,v\} \in E'} w(u,v)\len(I \cap [f(u),f(v)])\big)^2}{w(E')}\\
&\geq& \frac{(1-\beta)\len(I)^2 \cdot \phi(f)^2 \cdot \vol_f^2(a)}{\vol_f(I)}.
\end{eqnarray*}
Choosing $\beta=(\phi(f) \cdot \vol_f(a))/(\phi(f) \cdot \vol_f(a) + \vol_f(I))$ such that the above two terms are equal gives the lemma.
\end{proof}

We note that \autoref{l:drop} can be used to give a new proof of Cheeger's inequality with a weaker constant; see \autoref{s:cheeger}.

\section{Analysis of Spectral Partitioning}

Throughout this section we assume that $f\in \ell^2(V,w)$ is a non-negative function of norm $\norm{f}_w^2=1$ such that $\cR(f)\leq \lambda_2$ and $\vol(\supp(f))\leq \vol(V)/2$.  The existence of this function follows from \autoref{cor:smallsuppfunction}. In \autoref{subsec:cuhk}, we give our first proof of
\autoref{t:main} which is based on the idea of approximating $f$ by a $2k+1$ step function $g$. 
Our second proof is given in \autoref{subsec:stanford}.

\subsection{First Proof} \label{subsec:cuhk}

We say a function $g\in \ell^2(V,w)$ is a $l$-step approximation of $f$, if
there exist $l$ thresholds $0=t_0\leq t_1\leq \ldots\leq t_{l-1}$ such that for every vertex $v$,
$$ g(v) = \psi_{t_0,t_1,\ldots,t_{l-1}}(f(v)).$$ 
In words, $g(v)=t_i$ if $t_i$ is the closest threshold to $f(v)$; see \autoref{f:step} for an example.

\begin{figure*}[h] 
	\center
	\includegraphics[scale=0.6]{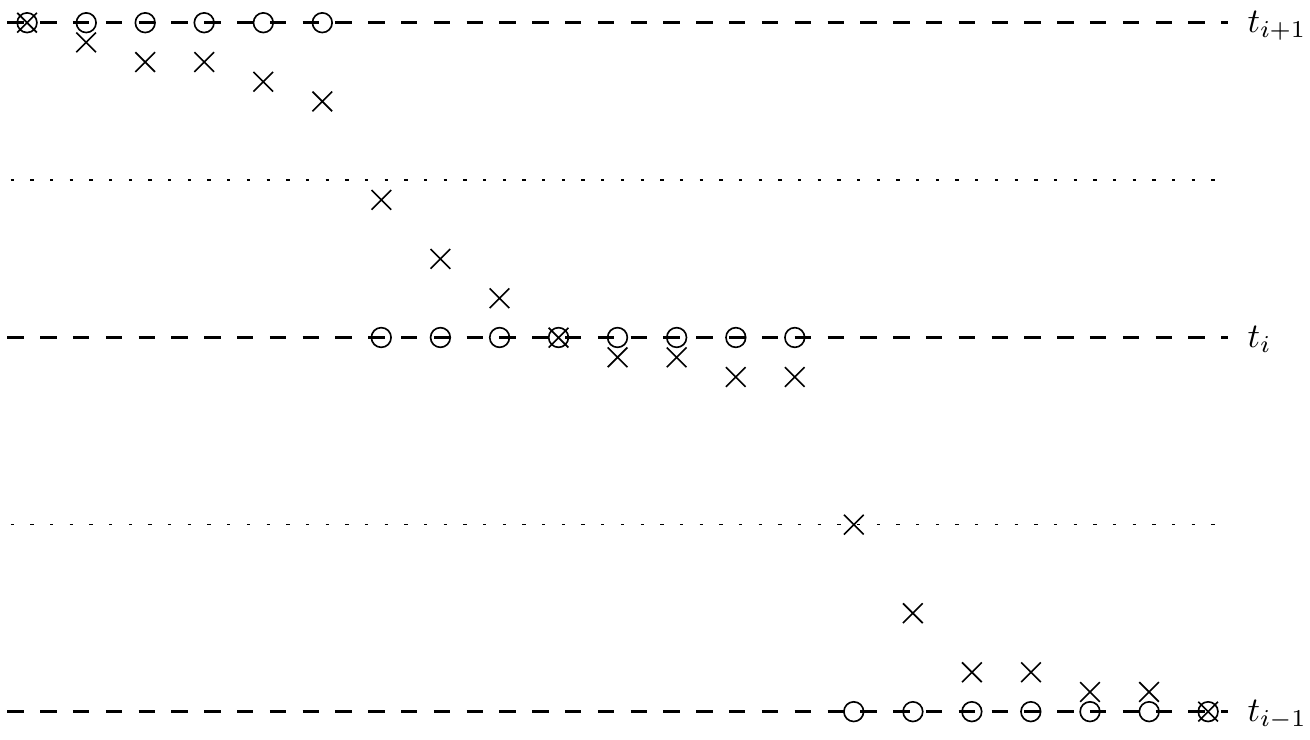}
	\caption{
The crosses denote the values of function $f$, and the circles denote the values of function $g$.
	}
\label{f:step}
\end{figure*}

We show that if there is a large gap between $\l_2$ and $\l_k$,
then the function $f$ is well approximated by a step function $g$ with at most $2k+1$ steps.
Then we define an appropriate $h$ and apply \autoref{lem:dirichletcheeger} to get a lower bound on the energy of $f$ in terms of $\norm{f-g}_w^2$.
One can think of $h$ as a probability distribution function on the threshold sets, and we will define $h$ in such a way that the threshold sets that are further away from the thresholds $t_0, t_1, \ldots, t_{2k}$ have higher probability.

\subsubsection*{Approximating $f$ by a $2k+1$ Step Function}
In the next lemma we show that if there is a large gap between $\cR(f)$ and $\lambda_k$, then there is a $2k+1$-step function $g$ such that
$\norm{f-g}_w^2 = O(\cR(f)/\lambda_k)$.

\begin{lemma} \label{l:lambda_k}
There exists a $2k+1$-step approximation of $f$, call $g$, such that
\begin{equation}
\label{eq:fgclose}
\norm{f-g}_w^2 \leq \frac{4\cR(f)}{\lambda_k}.
\end{equation}
\end{lemma}
\begin{proof}
Let $M:=\max_v f(v)$. We will find $2k+1$ thresholds $0=:t_0\leq t_1\leq \ldots\leq t_{2k} = M$, then
we let $g$ be a $2k+1$ step approximation of $f$ with these thresholds. 
Let $C:=2\cR(f)/k\lambda_k$.
We choose these thresholds inductively. 
Given $t_0, t_1, \ldots, t_{i-1}$, we let $t_{i-1}\leq t_i\leq M$ to be the smallest number such that 
\begin{equation}
\label{eq:equalnorm}
\sum_{v:t_{i-1} \leq f(v) \leq t_i} w(v) |f(v)-\psi_{t_{i-1},t_i}(f(v))|^2 = C.
\end{equation}
Observe that the left hand side varies continuously with $t_i$: 
when $t_i=t_{i-1}$ the left hand side is zero,
and for larger $t_i$ it is non-decreasing.
If we can satisfy  \eqref{eq:equalnorm} for some $t_{i-1}\leq t_i\leq M$,
then we let $t_i$ to be the smallest such number, and otherwise we set $t_i=M$. 

We say the procedure {\em succeeds} if $t_{2k}=M$. 
We will show that: (i) if the procedure succeeds then the lemma follows, 
and (ii) that the procedure always succeeds.
Part (i) is clear because if we define $g$ to be the $2k+1$ step approximation of $f$ with respect to $t_0,\ldots,t_{2k}$, then
$$\norm{f-g}_w^2 = \sum_{i=1}^{2k} \sum_{v:t_{i-1}\leq f(v)\leq t_i} w(v)|f(v)-\psi_{t_{i-1},t_i}(f(v))|^2   \leq 2kC = \frac{4\cR(f)}{\lambda_k},$$ 
 and we are done.  The inequality in the above equation follows by \eqref{eq:equalnorm}.

Suppose to the contrary that the procedure does not succeed.
We will construct $2k$ disjointly supported functions of Rayleigh quotients less than $\lambda_k/2$,
and then use \autoref{c:orthogonal} to get a contradiction. 
For $1\leq i\leq 2k$, let $f_i$ be the following function (see \autoref{f:disjoint} for an illustration):
\[
f_i(v) := \left\{
\begin{array}{ll}
|f(v)-\psi_{t_{i-1},t_i}(f(v))| & \mbox{if $t_{i-1} \leq f(v) \leq t_i$}\\
0 & \mbox{otherwise}.
\end{array} \right.
\]

\begin{figure}[h!] 
	\center
	\includegraphics[scale=0.6]{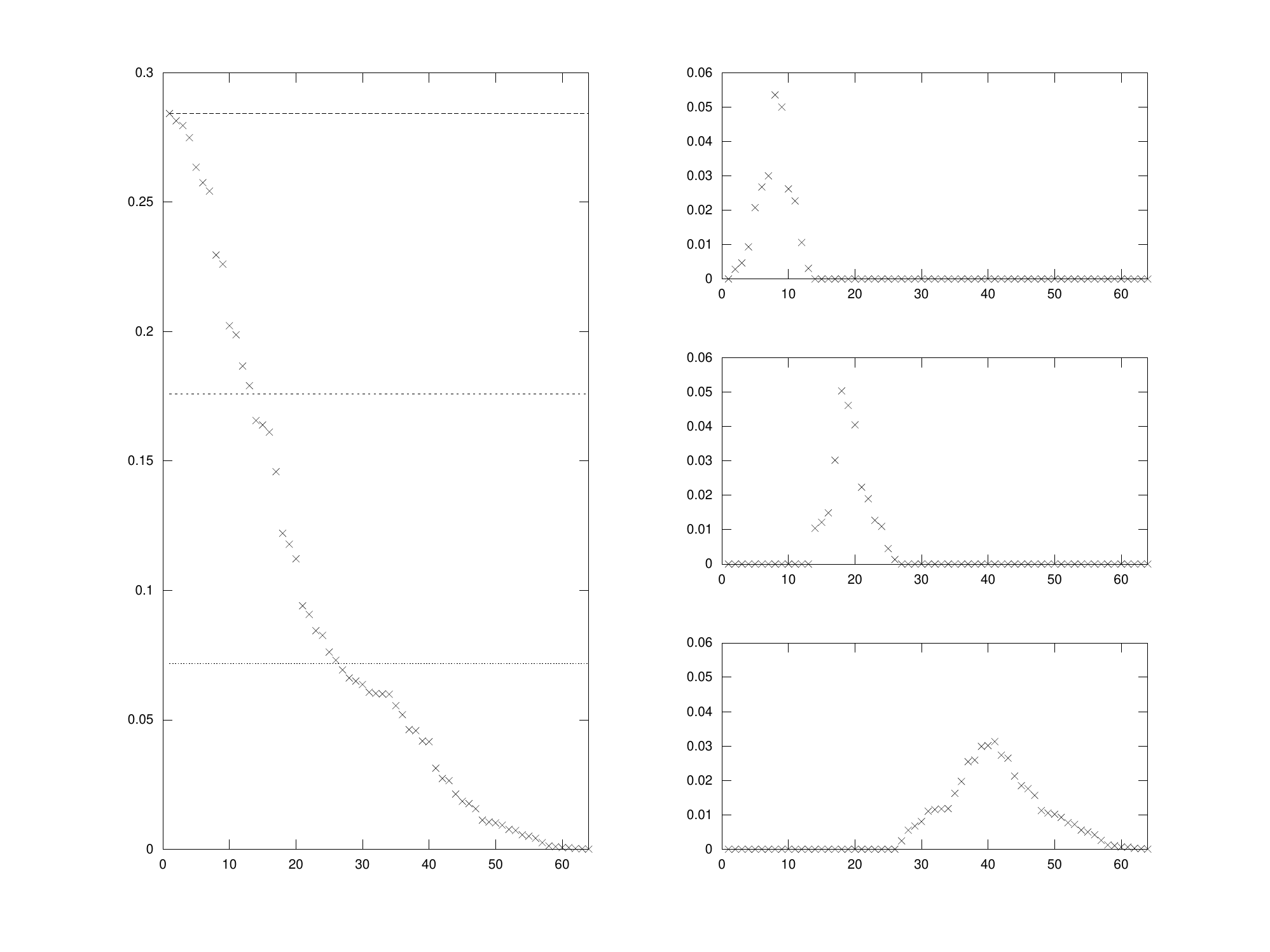}
	\caption{The figure on the left is the function $f$ with $\|f\|_w=1$.  
We cut $f$ into three disjointly supported vectors $f_1,f_2,f_3$ by
setting $t_0=0$, $t_1\approx 0.07$, $t_2 \approx 0.175$, and $t_3=\max f(v)$.
For each $1 \leq i \leq 3$, we define $f_i(v) = \min\{|f(v)-t_{i-1}|,|f(v)-t_i|\}$, if $t_{i-1} \leq f(v) \leq t_i$, and zero otherwise.} \label{f:disjoint}
\end{figure}

We will argue that at least $k$ of these functions have $\cR(f_i) < \frac 12 \lambda_k$.
By \eqref{eq:equalnorm}, we already know that the denominators of $\cR(f_i)$
are equal to $C$ ($\norm{f_i}_w^2=C$), so it remains to find an upper bound for the numerators. 
For any pair of vertices $u,v$, we show that
\begin{equation} 
\label{eq:lipschitzfi}
\sum_{i=1}^{2k} |f_i(u) - f_i(v)|^2 \leq |f(u)-f(v)|^2. 
\end{equation}
The inequality follows using the fact that $f_1,\ldots,f_{2k}$ are disjointly supported, and thus $u,v$ are contained in the support of at most two of these functions.
If both $u$ and $v$ are in the support of only one function, then \eqref{eq:lipschitzfi} holds since each $f_i$ is $1$-Lipschitz w.r.t. $f$. 
Otherwise, say $u\in \supp(f_i)$ and $v\in \supp(f_j)$ for $i<j$, 
then \eqref{eq:lipschitzfi} holds since
\begin{eqnarray*}
|f_i(u) - f_i(v)|^2 + |f_j(u) - f_j(v)|^2 &=& |f(u) - g(u)|^2 + |f(v)-g(v)|^2\\
& \leq& |f(u) - t_i|^2 + |f(v) - t_i|^2 \leq |f(u)-f(v)|^2.
\end{eqnarray*}

Summing \eqref{eq:lipschitzfi} we have
\[\sum_{i=1}^{2k} \cR(f_i) = \frac{1}{C} \sum_{i=1}^{2k} \sum_{u\sim v} w(u,v)|f_i(u) - f_i(v)|^2 \leq \frac{1}{C} \sum_{u\sim v} w(u,v)|f(u)-f(v)|^2 = \frac{k\lambda_k}{2}.\]
Hence, by an averaging argument, there are $k$ disjointly functions $f'_1,\ldots,f'_k$ of Rayleigh quotients less than $\lambda_k/2$,
a contradiction to \autoref{c:orthogonal}. 
\end{proof}


\subsubsection*{Upper Bounding $\phi(f)$ Using $2k+1$ Step Approximation $g$}

Next , we show that we can use any function $g$ that is a $2k+1$ approximation of $f$ to upper-bound
$\phi(f)$ in terms of $\norm{f-g}_w$.

\begin{proposition}
\label{l:jump}
For any $2k+1$-step approximation of $f$ with $\norm{f}_w=1$, called $g$, 
\[ \phi(f) \leq 4k\cR(f) + 4\sqrt2k \norm{f-g}_w \sqrt{ \cR(f) }.   \]
\end{proposition}

%
Let $g$ be a $2k+1$ approximation of $f$ with thresholds $0=t_0\leq t_1\leq \ldots\leq t_{2k}$, i.e. $g(v) := \psi_{t_0,t_1,\ldots,t_{2k}}(f(v))$.
We will define a function $h\in \ell^2(V,w)$ such that each threshold set of $h$ is also a threshold set of $f$ (in particular $\supp(h) = \supp(f)$), and
\begin{equation}
 \frac {\sum_{u\sim v} w(u,v) |h(v) - h(u)| }{\sum_v w(v) h(v) }  
\leq 4k\cR(f) + 4\sqrt2k \norm{f-g}_w \sqrt{ \cR(f) }.
\label{eq:hrayleigh}
\end{equation}
We then simply use \autoref{lem:dirichletcheeger} to prove \autoref{l:jump}.

Let  $\mu:\R\rightarrow\R$,
\[ \mu(x) := |x - \psi_{t_0,t_1,\ldots,t_{2k}}(x)|. \] 
Note
that $|f(v) - g(v)| = \mu(f(v))$.
One can think of $\mu$ as a probability density function to sample the threshold sets, where threshold sets that are further away from the thresholds $t_0, t_1, \ldots, t_{2k}$ are given higher probability.
We define $h$ as follows:

\[ h(v) := \int_{0}^{f(v)} \mu(x) dx \]

Observe that the threshold sets of $h$ and the threshold sets of $f$ are the same,
as $h(u) \geq h(v)$ if and only if $f(u) \geq f(v)$.
It remains to prove \eqref{eq:hrayleigh}. 
We use the following two claims, that bound the denominator and the numerator separately. 
%

\begin{claim} 
\label{cl:denom}
For every vertex $v$,
\[ h(v) \geq \frac 1{8k}  f^2(v).  \]
\end{claim}

\begin{proof} If $f(v)=0$, then $h(v)=0$ and there is nothing to prove. Suppose $f(v)$ is in
the interval $f(v)\in [t_i,t_{i+1}]$. 
Using the Cauchy-Schwarz inequality,
\[ f^2(v) = (\sum_{j=0}^{i-1} (t_{j+1}-t_j ) + (f(v) - t_i) )^2 \leq 2k \cdot ( \sum_{j=0}^{i-1} (t_{j+1}-t_j )^2
+ (f(v) - t_i)^2 ).\]
On the other hand, by the definition of $h$,

\begin{eqnarray*} h(v) &=& \sum_{j=0}^{i-1} \int_{t_j}^{t_{j+1}} \mu(x)dx + \int_{t_i}^{f(v)} \mu(x)dx \\
&=& \sum_{j=0}^{i-1} \frac14 (t_{j+1}-t_j)^2 + \int_{t_i}^{f(v)} \mu(x)dx 
\geq \sum_{j=0}^{i-1} \frac14 (t_{j+1}-t_j)^2 + \frac14(f(v)-t_i)^2,
\end{eqnarray*}
where the inequality follows by the fact that $f(v)\in [t_i,t_{i+1}]$.
\end{proof}

And we will bound the numerator with the following claim.

\begin{claim} 
\label{cl:numerator}
For any pair of vertices $u,v\in V$,
\[ | h(v)- h(u) | \leq \frac12 |f(v) - f(u) | \cdot ( |f(u) - g(u) | + |f(v)-g(v)| + |f(v) - f(u) | ). \]
\end{claim}
\begin{proof}
By the definition of $\mu(.)$, for any $x\in [f(u),f(v)]$,
\begin{eqnarray*} \mu(x) \leq \min\{ |x-g(u)|, |x-g(v)|\} &\leq& \frac{|x-g(u)| + |x-g(v)|}2 \\
&\leq&  \frac12 \Big( (|x-f(u)| + |f(u) - g(u)|) + (|x-f(v)|+|f(v)-g(v)|) \Big)\\
&=& \frac12( |f(u) - g(u) | + |f(v)-g(v)| + |f(v) - f(u) | ),
\end{eqnarray*}
where the third inequality follows by the triangle inequality, and the last equality uses $x\in [f(u),f(v)]$. 
Therefore, 
\begin{eqnarray*} h(v)-h(u) = \int_{f(u)}^{f(v)} \mu(x) dx& \leq &|f(v)-f(u)| \cdot \max_{x\in [f(u),f(v)]} \mu(x) \\
&\leq& \frac12 |f(v)-f(u)| \cdot (|f(u) - g(u) | + |f(v)-g(v)| + |f(v) - f(u) |). 
\end{eqnarray*}
\end{proof} 

Now we are ready to prove \autoref{l:jump}.

\begin{proofof}{\autoref{l:jump}}
First, by \autoref{cl:numerator},
\begin{align*} 
\sum_{u\sim v} w(u,v)|h(u)-h(v)| &\leq  \sum_{u\sim v} \frac12 w(u,v)|f(v) - f(u) | \cdot ( |f(u) - g(u) | + |f(v)-g(v)| + |f(v) - f(u) | ) \\
& \leq  \frac12 \cR(f) + \frac12 \sqrt{\sum_{u\sim v} w(u,v)|f(v) - f(u) |^2}  \sqrt { \sum_{u\sim v} w(u,v) (|f(u) - g(u) | + |f(v)-g(v)|)^2} \\
& \leq \frac12 \cR(f) + \frac12 \sqrt{\cR(f)} \cdot \sqrt{2  \sum_{u\sim v} w(u,v) (|f(u) -g(u) |^2 + |f(v)-g(v)|^2)} \\
 &=  \frac12 \cR(f) + \frac12 \sqrt{\cR(f)} \cdot \sqrt{2\norm{f-g}^2_w},
\end{align*}
where the second inequality follows by the Cauchy-Schwarz inequality. 
On the other hand, by \autoref{cl:denom},

\[ \sum_v w(v) h(v) \geq \frac 1{8k} \sum_v w(v) f^2(v) = \frac 1 {8k} \norm{f}^2_w = \frac 1 {8k}. \]

Putting above equations together proves \eqref{eq:hrayleigh}. Since the threshold
sets of $h$ are the same as the threshold sets of $f$, 
we have $\phi(f)=\phi(h)$ and the
proposition follows by \autoref{lem:dirichletcheeger}.
\end{proofof}

Now we are ready to prove \autoref{t:main}.
\begin{proofof}{\autoref{t:main}}
Let $g$ be as defined in \autoref{l:lambda_k}. 
By \autoref{l:jump}, we get
\[
\phi(f) \leq 4k\cR(f) + 4\sqrt2k \norm{f-g}_w \sqrt{ \cR(f) }
\leq 4k\cR(f) + 8\sqrt2 k \cR(f) / \sqrt{\l_k} \leq 12\sqrt{2} k \cR(f) / \sqrt{\l_k}.
\]
\end{proofof}

We provide a different proof of \autoref{t:main} in \autoref{sec:jumpsecond} by
lower-bounding $\cE_f$ using a $2k+1$ approximation of $f$. 
This proof uses \autoref{l:drop} instead of \autoref{lem:dirichletcheeger} to prove the theorem. 

{\bf Remark:}
\autoref{cl:numerator} can be improved to 
\[ | h(v)- h(u) | \leq \frac12 |f(v) - f(u) | \cdot ( |f(u) - g(u) | + |f(v)-g(v)| + \frac12 |f(v) - f(u) | ), \]
and thus \autoref{t:main} can be improved to $\phi(f) \leq 10\sqrt{2}k \cR(f)/\sqrt{\l_k}$.

\subsection{Second Proof} \label{subsec:stanford}

Instead of directly proving \autoref{t:main}, we  use \autoref{cor:smallsuppfunction} and \autoref{c:orthogonal} and prove a stronger version,
 as it will be used later to prove \autoref{c:multiway}\footnote{We note that the first proof can also be modified to obtain this stronger version, without the additional property that each $f_i$
 is defined on an interval $[a_i,b_i]$ of the form $|a_i-b_i|=\Theta(1/k)a_i$.
See \autoref{lem:ksteps} for an adaptation of \autoref{l:lambda_k} to prove such a statement for maximum cut.}.
In particular, instead of directly upper-bounding $\lambda_k$, we construct $k$ disjointly supported functions with small Rayleigh quotients. 
\begin{theorem}
\label{t:mainfunc}
For any non-negative function $f\in \ell^2(V,w)$ such that $\supp(f)\leq \vol(V)/2$, and  $\delta:=\phi^2(f)/\cR(f)$, at least one of the following holds
\begin{enumerate}[i)]
\item $\phi(f)\leq O(k) \cR(f);$
\item There exist $k$ disjointly supported functions $f_1,f_2,\ldots,f_k$ such that for all $1\leq i\leq k$, $\supp(f_i)\subseteq \supp(f)$ and
$$  \cR(f_i) \leq O(k^2) \cR(f)/\delta.$$
Furthermore, the support of each $f_i$ is an interval $[a_i,b_i]$ such that
$|a_i-b_i| = \Theta(1/k) a_i$. 
\end{enumerate}

\end{theorem}
We will show that if $\cR(f) = \Theta(\phi(G)^2)$ (when $\delta = \Theta(1)$), 
then $f$ is a smooth function of the vertices, 
in the sense that in any interval of the form $[t, 2t]$ we expect
the vertices to be embedded in equidistance positions. It is instructive to verify this for the second eigenvector of the cycle.

\subsubsection*{Construction of Disjointly Supported Functions Using Dense Well Separated Regions}

First, we show that \autoref{t:mainfunc} follows from a construction of $2k$ dense well separated regions, and in the subsequent parts we construct these regions based on $f$.
A region $R$ is a closed subset of $\mathbb{R}_+$. 
Let $\mass(R):=\sum_{v: f(v)\in R} w(v) f^2(v)$.
We say $R$ is {\em $W$-dense} if $\mass(R)\geq W$.
For any $x\in \mathbb{R}_+$, we define
$$\dist(x,R):=\inf_{y\in R} \frac{|x-y|}{y}.$$
The {\em $\eps$-neighborhood} of a region $R$ is the set of points at distance at most $\eps$ from $R$, 
$$N_{\eps}(R) :=\{ x\in\mathbb{R}_+: \dist(x,R) < \eps \}.$$
 We say two regions $R_1,R_2$ are {\em $\eps$-well-separated}, if $N_{\eps}(R_1)\cap N_{\eps}(R_2)=\emptyset$. 
In the next lemma, we show that our main theorem can be proved by finding $2k$, $\Omega(\delta/k)$-dense, $\Omega(1/k)$ well-separated regions. 

\begin{lemma}
\label{lem:densewellsep}
Let $R_1,R_2,\ldots,R_{2k}$ be a set of $W$-dense and $\eps$-well separated regions. 
Then, there are $k$ disjointly supported functions $f_1,\ldots,f_k$, each supported on the $\eps$-neighborhood of one of the regions such that
$$\forall~ 1\leq i\leq k,~\cR(f_i) \leq \frac{2\cR(f)}{k\eps^2 W}.$$
\end{lemma}
\begin{proof}
For any $1\leq i\leq 2k$, we define a function $f_i$, where for all $v\in V$,
$$ f_i(v):=f(v)\max\{0, 1-\dist(f(v),R_i)/\eps\}.$$
Then, $\norm{f_i}_w^2 \geq \mass(R_i)$. Since the regions are $\eps$-well separated, the functions are disjointly supported. Therefore,
the endpoints of each edge $\{u,v\}\in E$ are in the support of at most two functions. 
Thus, by an averaging argument, there exist $k$ functions $f_1,f_2,\ldots,f_k$ (maybe after renaming) satisfy the following. For all $1\leq i\leq k$,
$$ \sum_{u\sim v} w(u,v)|f_i(u) - f_i(v)|^2 \leq \frac{1}{k} \sum_{j=1}^{2k} \sum_{u\sim v} w(u,v)|f_j(u) - f_j(v)|^2.$$
Therefore, for $1\leq i\leq k$,
\begin{eqnarray*} \cR(f_i) = \frac{\sum_{u\sim v} w(u,v)|f_i(u) - f_i(v)|^2}{\norm{f_i}_w^2} &\leq& \frac{\sum_{j=1}^{2k}\sum_{u\sim v} w(u,v)|f_j(u) - f_j(v)|^2}{k \cdot \min_{1\leq i\leq 2k} \norm{f_i}_w^2} \\
&\leq& \frac{2\sum_{u\sim v} w(u,v)|f(u)-f(v)|^2}{k\eps^2 W} = \frac{2\cR(f)}{k\eps^2W},
\end{eqnarray*}
where we used the fact that  $f_j$'s are $1/\eps$-Lipschitz. Therefore, $f_1,\ldots,f_k$ satisfy  lemma's statement. 
\end{proof}

\subsubsection*{Construction of Dense Well Separated Regions}

Let $0<\alpha<1$ be a constant that will be fixed later in the proof.
For $i\in\mathbb{Z}$, we define the interval $I_i:=[\alpha^i,\alpha^{i+1}]$. Observe that these intervals partition the vertices with positive value in $f$.
We  let $\mass_i := \mass(I_i)$.
We partition each interval $I_i$ into $12k$ subintervals of equal length,
$$I_{i,j} := \left[\alpha^i \left(1-\frac{j(1-\alpha)}{12k}\right), \alpha^i\left(1-\frac{(j+1)(1-\alpha)}{12k}\right)\right],$$
for $0\leq j<12k$. Observe that for all $i,j$,
\begin{equation} \len(I_{i,j}) = \frac{\alpha^{i}(1-\alpha)}{12k}.  
\label{eq:subintervallength}
\end{equation}
Similarly we define $\mass_{i,j} := \mass(I_{i,j})$. 
We say a subinterval $I_{i,j}$ is {\em heavy}, if $\mass_{i,j} \geq c\delta \mass_{i-1} /   k$, where $c>0$ is a constant that will be fixed later in the proof; we say it is {\em light} otherwise. We use $H_i$ to denote the set of heavy subintervals of $I_i$ and $L_i$ for the set of light subintervals. 
We use $h_i$ to denote the number of heavy subintervals. 
We also say an interval $I_i$ is {\em balanced} if $h_i \geq 6k$, 
denoted by $I_i \in B$ where $B$ is the set of balanced intervals. 
Intuitively, an interval $I_i$ is {\em balanced} if the vertices are distributed uniformly inside that interval.

Next we describe our proof strategy.
Using \autoref{lem:densewellsep} to prove the theorem it is sufficient to find $2k$, $\Omega(\delta/k)$-dense, $\Omega(1/k)$ well-separated regions $R_1,\ldots,R_{2k}$. 
Each of our $2k$ regions will be a union of {\em heavy} subintervals. 
Our construction is simple: from each balanced interval we choose $2k$ {\em separated} heavy subintervals and include each of them in one of the regions. In order to promise that the regions are well separated, once we include $I_{i,j}\in H_i$ into a region $R$ we leave the two neighboring subintervals $I_{i,j-1}$ and $I_{i,j+1}$ unassigned, 
so as to separate $R$ from the rest of the regions. 
In particular, for all $1\leq a\leq 2k$ and all $I_i\in B$, 
we include the $(3a-1)$-th heavy subinterval of $I_i$ in $R_a$. 
Note that if an interval $I_i$ is balanced, then it has $6k$ heavy subintervals and we can include one heavy subinterval in each of the $2k$ regions.
Furthermore, by \eqref{eq:subintervallength}, 
the regions are $(1-\alpha)/12k$-well separated. 
 It remains to prove that these $2k$ regions are dense. 
Let $$ \Delta:=\sum_{I_{i} \in B} \mass_{i-1}$$
be the summation of the mass of the preceding interval of balanced intervals. 
Then, since each heavy subinterval $I_{i,j}$ has a mass of $c\delta \mass_{i-1}/ k$, by the above construction all regions are $c\Delta \delta / k$-dense.
Hence, the following proposition follows from \autoref{lem:densewellsep}.

\begin{proposition}
\label{prop:lambdakbound}
There are $k$ disjoint supported functions $f_1,\ldots,f_k$ such that for all $1\leq i\leq k$, $\supp(f_i)\subseteq \supp(f)$ and
$$ \forall~ 1\leq i\leq k,~ \cR(f_i) \leq \frac{300 k^2 \cR(f)}{(1-\alpha)^2c\delta \Delta}.$$
\end{proposition}

\subsubsection*{Lower Bounding the Density}

So in the rest of the proof we just need to lower-bound $\Delta$ by an absolute constant.

\begin{proposition}
\label{lem:unbalancedinterval}
For any interval $I_i\notin B$,
$$ \cE(I_i) \geq \frac{\alpha^6\phi(f)^2\mass_{i-1}(1-\alpha)^2}{24 ( k\alpha^4 \phi(f)+c\delta)}.
$$
\end{proposition}
\begin{proof}
In the next claim, 
we lower-bound the energy of a light subinterval in terms of $\mass_{i-1}$. 
Then, we prove the statement simply using $h_i < 6k$.
\begin{claim}
\label{lem:lightsubinterval}
For any light subinterval $I_{i,j}$,
$$ \cE(I_{i,j}) \geq \frac{\alpha^6\phi(f)^2\mass_{i-1}(1-\alpha)^2}{144 k ( k \alpha^4 \phi(f) +c\delta)}.
$$
\end{claim}
\begin{proof}
First, observe that
\begin{equation}
\label{eq:dpreinterval}
 \mass_{i-1} = \sum_{v\in I_{i-1}} w(v) f^2(v) \leq \alpha^{2i-2} \vol(\alpha^i).
 \end{equation}
Therefore,
\begin{equation}
\label{eq:energyJ} 
\vol(I_{i,j}) = \sum_{v \in I_{i,j}} w(v)  \leq  \sum_{v\in I_{i,j}}  w(v) \frac{f^2(v)}{\alpha^{2i+2}} = \frac{\mass_{i,j}}{\alpha^{2i+2}} \leq  \frac{c\delta \mass_{i-1}}{ k\alpha^{2i+2}} \leq \frac{c\delta \vol(\alpha^{i})}{ k\alpha^4}, 
\end{equation}
where we use the assumption that $I_{i,j}\in L_i$ in the second last inequality, and \eqref{eq:dpreinterval} in the last inequality. 
By \autoref{l:drop},
$$\cE(I_{i,j}) \geq \frac{\phi(f)^2 \cdot \vol(\alpha^i)^2 \cdot \len(I_{i,j})^2}{\phi(f) \cdot \vol(\alpha^i)+\vol(I_{i,j})} 
\geq \frac{k \alpha^4 \phi(f)^2 \cdot \vol(\alpha^i) \cdot \len(I_{i,j})^2}{k \alpha^4 \phi(f) + c\delta}
\geq \frac{\alpha^6\phi(f)^2\mass_{i-1}(1-\alpha)^2}{144 k (k \alpha^4 \phi(f)+c\delta)},$$
where the first inequality holds by \eqref{eq:energyJ}, and the last inequality holds by \eqref{eq:subintervallength} and \eqref{eq:dpreinterval}.
\end{proof}

Now, since the subintervals are disjoint, by \autoref{f:additivity},
\begin{eqnarray*}
\cE(I_i) \geq \sum_{I_{i,j}\in L_i}\cE(I_{i,j})& 
\geq & (12k-h_i) \frac{\alpha^6\phi(f)^2\mass_{i-1}(1-\alpha)^2}{144 k (k\phi(f) \alpha^4+c\delta)} 
\geq \frac{\alpha^6\phi(f)^2\mass_{i-1}(1-\alpha)^2}{24 ( k\phi(f) \alpha^4+c\delta)},
 \end{eqnarray*}
where we used the assumption that $I_i$ is not balanced and thus $h_i < 6k$.
\end{proof}

Now we are ready to lower-bound $\Delta$.

\begin{proofof}{\autoref{t:mainfunc}}
First we show that $\Delta \geq 1/2$, unless (i) holds, and then we use \autoref{prop:lambdakbound} to prove the theorem.  
If $\phi(f)\leq 10^4 k\cR(f)$, then (i) holds and we are done. 
So, assume that
\begin{equation}
\label{eq:lambdakbound}
 \frac{10^8 k^2\cR^2(f)}{\phi^2(f)} \leq 1,
 \end{equation}
and we prove (ii).
Since $\norm{f}_w^2=1$, by \autoref{lem:unbalancedinterval},
$$\cR(f) = \cE_f \geq \sum_{I_{i}\notin B} \cE(I_i)\geq \sum_{I_i\notin B}
\frac{\alpha^6\phi(f)^2\mass_{i-1}(1-\alpha)^2}{24 (k\phi(f) \alpha^4+c\delta)}
.$$
Set $\alpha=1/2$ and $c:=\alpha^6(1-\alpha)^2/96$.
If $k\phi(f) \alpha^4\geq c\delta$, then we get
$$ \sum_{I_i\notin B} \mass_{i-1} \leq \frac{48k\cR(f)}{\alpha^2(1-\alpha)^2\phi(f)} \leq \frac{1}{2},$$
where the last inequality follows from \eqref{eq:lambdakbound}. 
Otherwise, 
\[\sum_{I_i \notin B} \mass_{i-1} \leq \frac{48c\delta \cR(f)}{\alpha^6 (1-\alpha)^2 \phi^2(f)} \leq \frac12,\]
where the last inequality follows from the definition of $c$ and $\delta$.
Since $\mass(V)=\norm{f}_w^2 = 1$, it follows from the above equations that 
$ \Delta \geq \frac{1}{2}.$
Therefore, by \autoref{prop:lambdakbound}, we get $k$ disjointly supported functions
$f_1,\ldots,f_k$ such that
$$ \cR(f_i) \leq  \frac{300 k^2 \cR(f)}{(1-\alpha)^2c\delta \Delta}
\leq \frac{10^8 k^2 \cR(f)^2}{\phi(f)^2}.
$$
Although each function $f_i$ is defined on a region which is a union of many heavy subintervals,
we can simply restrict it to only one of those subintervals guaranteeing that $\cR(f_i)$
only decreases. Therefore  each $f_i$ is defined on an interval $[a_i,b_i]$
where by \eqref{eq:subintervallength}, $|a_i-b_i| = \Theta(1/k) a_i$.
This proves (ii).
\end{proofof}


%
%
\section{Extensions and Connections} \label{s:extensions}

In this section, we extend our approach to other graph partitioning problems,
including multiway partitioning (\autoref{s:multiway}), balanced separator (\autoref{s:bisection}), maximum cut (\autoref{s:maxcut}), and to the manifold setting (\autoref{s:manifold}).
Also, we discuss some relations between our setting and the settings for planted and semirandom instances (\autoref{s:semirandom}) and in stable instances (\autoref{s:stable}).

\subsection{Spectral Multiway Partitioning} \label{s:multiway}

In this subsection,
we use \autoref{t:mainfunc} and the results in~\cite{lee-gharan-trevisan} 
to prove \autoref{c:multiway}.

\begin{theorem}[{\cite[Theorem 1.3]{lee-gharan-trevisan}}]
For any graph $G=(V,E,w)$ and any integer $k$, there exist $k$ non-negative disjointly supported functions $f_1, \ldots, f_k \in \ell^2(V,w)$ such that for each $1 \leq i \leq k$ we have $\cR(f_i) \leq O(k^6) \l_k$.
\end{theorem}
Let $f_1,\ldots,f_k$ be as defined above. 
We consider two cases. First assume that $\vol(\supp(f_i))\leq \vol(V)/2$ for all $1\leq i\leq k$. Recall that $V_{f_i}(t_\opt)$ is the best threshold set of $f_i$. Let $S_i:=V_{f_i}(t_\opt)$. Then, for each function $f_i$, by \autoref{t:mainfunc},
\[\phi(S_i) = \phi(f_i) \leq O(l) \frac{\cR(f_i)}{\sqrt{\lambda_l}} \leq O(l k^6) \frac{\l_k}{\sqrt{\l_l}}.\]
Furthermore, since $S_i\subseteq \supp(f_i)$ and $f_1,\ldots,f_k$ are disjointly supported, $S_1,\ldots,S_k$ are disjoint. 
Hence, 
$$\phi_k(G) = \max_{1\leq i\leq k} \phi(S_i)\leq O(lk^6) \frac{\lambda_k}{\sqrt{\lambda_l}},$$
and we are done.
Now suppose there exists a function, say $f_k$, with $\vol(\supp(f_k))>\vol(V)/2$. Let $S_i=V_{f_i}(t_\opt)$ for $1\leq i\leq k-1$, and $S_k:=V\setminus S_1\setminus \ldots \setminus S_{k-1}$. Similar to the above, the sets $S_1,\ldots,S_{k-1}$ are disjoint, and $\phi(S_i)\leq O(lk^6\lambda_k/\sqrt{\lambda_l})$ for all $1\leq i\leq k-1$. 
Observe that
$$ \phi(S_k)  = \frac{w(E(S_1,S_k)) + \ldots+w(E(S_{k-1},S_k))}{\vol(V)-\vol(S_k)}\leq \frac{\sum_{i=1}^{k-1} w(E(S_i,\overline{S_i}))}{\sum_{i=1}^{k-1} \vol(S_i)}\leq O(lk^6) \frac{\lambda_k}{\sqrt{\lambda_l}},$$
where the first equality uses $\vol(S_k)\geq \vol(V)/2$. Hence, $\phi_k(G)\leq O(lk^6)\lambda_k/\sqrt{\lambda_l}$. This completes the proof of (i) of \autoref{c:multiway}.

To prove (ii) we use the following theorem of \cite{lee-gharan-trevisan}.
\begin{theorem}[{\cite[Theorem 4.6]{lee-gharan-trevisan}}]
For any graph $G=(V,E,w)$ and $\delta>0$, the following holds: 
For any  $k\geq 2$, there exist $r\geq (1-\delta)k$
non-negative disjointly supported functions $f_1,\ldots,f_r\in \ell^2(V,w)$ such that
for all $1\leq i\leq r$,
$$ \cR(f_i)\leq O(\delta^{-7}\log^2 k) \lambda_k.$$
\end{theorem}
It follows from (i) that without loss of generality we can assume that $\delta>10/k$. Let $\delta':=\delta/2$. 
Then, by the above theorem, there exist $r\geq (1-\delta')k$ non-negative disjointly supported functions $f_1,\ldots,f_r$ such that $\cR(f_i)\leq O(\delta^{-7}\log^2k)\lambda_k$ and $\vol(\supp(f_i))\leq \vol(V)/2$. 
For each $1\leq i\leq r$, let $S_i:=V_{f_i}(t_\opt)$. 
Similar to the argument in part (i), since $S_i\subseteq \supp(f_i)$,  
the sets $S_1,\ldots,S_r$ are disjoint. 
Without loss of generality assume that $\phi(S_1)\leq \phi(S_2)\leq\ldots\phi(S_r)$. 
Since $S_1,\ldots,S_{(1-\delta)k}$ are disjoint, 
\begin{equation}
\label{eq:upperphideltak}
 \phi_{(1-\delta)k}(G) \leq \phi(S_{(1-\delta)k+1}) \leq \ldots\leq \phi(S_r).
 \end{equation}
Let $m:=l/(\delta' k)=2l/(\delta k)$. 
If $\phi(f_i)\leq O(m)\cR(f_i)$ for some $(1-\delta)k < i\leq r$, then we get
$$ \phi_{(1-\delta)k}(G) \leq \phi(S_i) = \phi(f_i) \leq O(m) \cR(f_i) \leq O\left(\frac{l\log^2 k }{\delta^8k}\right)\lambda_k,$$
and we are done.
Otherwise, by \autoref{t:mainfunc}, for each $(1-\delta)k < i\leq r$, there exist $m$ disjointly
supported functions $h_{i,1},\ldots h_{i,m}$ such that for all $1\leq j\leq m$,
$\supp(h_{i,j})\subseteq \supp(f_i)$ and
\begin{equation}
\label{eq:upperhij}
 \cR(h_{i,j}) \leq O(m^2) \frac{\cR(f_i)^2}{\phi(f_i)^2}\leq O\left(\frac{l^2}{\delta^2 k^2}\right) \frac{O(\delta^{-14}\log^4k)\lambda_k^2}{\phi^2_{(1-\delta)k}(G)} = O\left(\frac{l^2 \log^4 k}{\delta^{16}k^2}\right) \frac{\lambda_k^2}{\phi^2_{(1-\delta)k}(G)}
 \end{equation}
where the second inequality follows from \eqref{eq:upperphideltak}. Since $f_{(1-\delta)k+1},\ldots,f_r$ are disjointly supported, all functions $h_{i,j}$ are disjointly supported as well. 
Therefore, since $l=m(\delta' k) \leq m(r-(1-\delta)k)$, by \autoref{c:orthogonal},
$$ \lambda_l \leq 2\max_{\substack{(1-\delta)k<i\leq r\\1\leq j\leq m}} \cR(h_{i,j}) \leq O\left(\frac{l^2 \log^4 k}{\delta^{16}k^2}\right) \frac{\lambda_k^2}{\phi^2_{(1-\delta)k}(G)}, $$
where the second inequality follows from \eqref{eq:upperhij}. 
This completes the proof of (ii) of \autoref{c:multiway}.

Part (iii) can be proved in a very similar way to part (ii). We just exploit the following theorem
of \cite{lee-gharan-trevisan}.
\begin{theorem}[{\cite[Theorem 3.7]{lee-gharan-trevisan}}]
For any graph $G=(V,E,w)$ that excludes $K_h$ as a minor, and any $\delta\in(0,1)$,
there exists $(1-\delta)k$ non-negative disjointly supported functions
$f_1,\ldots,f_{(1-\delta)k}$ such that
$$ \cR(f_i)\leq O(h^4\delta^{-4}) \lambda_k.$$
\end{theorem}
We follow the same proof steps as in part (ii) except that we upper bound $\cR(f_i)$ by $O(h^4 \delta^{-4})\lambda_k$. 
This completes the proof of \autoref{c:multiway}.

In the remaining part of this section we describe some examples. 
First we show that there exists a graph where $\phi_k(G) \geq \Omega(l-k+1) \lambda_k/\sqrt{\lambda_l}$. Let $G$ be a union of $k-2$ isolated vertices and a cycle of length $n$.
Then, $\phi_k(G) = \Theta(1/n)$, $\lambda_k = \Theta(1/n^2)$ and for $l>k$, $\lambda_{l} = \Theta((l-k+1)^2/n^2)$.
Therefore,
$$ \phi_k(G) \geq \Omega(l-k+1) \frac{\lambda_k}{\sqrt{\lambda_l}}$$
The above example shows that for $l\gg k$, the dependency on $l$ in the right hand
side of part (i) of \autoref{c:multiway} is necessary. 

Next we show that there exists a graph where $\phi_{k/2}(G) \geq \Omega(l/k) \lambda_k/\sqrt{\lambda_l}$. 
Let $G$ be a cycle of length $n$. Then, $\phi_{k/2}(G) = \Theta(k/n)$, $\lambda_k = \Theta(k^2/n^2)$ and $\lambda_l = \Theta(l^2/n^2)$. Therefore,
$$\phi_{k/2}(G) \geq \Omega(l/k) \frac{\lambda_k}{\sqrt{\lambda_l}}.$$
This shows that part (iii) of \autoref{c:multiway} is tight (up to constant factors) when $\delta$ is a constant.

\subsection{Balanced Separator} \label{s:bisection}

In this section we give a simple polynomial time algorithm with approximation factor $O(k/\lambda_k)$ for the balanced separator problem.
We restate \autoref{thm:minbisection-intro} as follows.

\begin{theorem}
\label{thm:minbisection}
Let 
$$\eps:=\min_{\vol(S)=\vol(V)/2} \phi(S).$$
There is a polynomial time algorithm
that finds a set $S$ such that $\frac15 \vol(V)\leq \vol(S)\leq \frac45 \vol(V)$, and $\phi(S)\leq O(k \eps/\lambda_k)$.
\end{theorem}

We will prove the above theorem by repeated applications of \autoref{t:main}.
Our algorithm is similar to the standard algorithm for finding a balanced separator by applying Cheeger's inequality repeatedly.
We inductively remove a subset of vertices of the remaining graph such that the union of the removed vertices is a non-expanding set in $G$, until the set of removed vertices has at least a quarter of the total volume. 
The main difference is that besides removing a sparse cut by applying \autoref{t:main}, there is an additional step that removes a subset of vertices such that the conductance of the union of the removed vertices does not increase.
The details are described in \autoref{alg:minbisection}.

\begin{algorithm}
\begin{algorithmic}
\STATE $U\leftarrow V$.
\WHILE {$\vol(U) > \frac45 \vol(V)$}
\STATE Let $H=(U,E(U))$ be the induced subgraph of $G$ on $U$, and $\lambda'_2$ be the second smallest eigenvalue of $\cL_H$.
\STATE Let $f\in \ell^2(U,w)$ be a non-negative function such that $\vol(\supp(f))\leq \vol(H)/2$, and $\cR_H(f)\leq \lambda'_2$.
\IF{$\phi_H(f) \leq O(k)\cR_H(f)/\sqrt{\lambda_k}$}
\STATE $U \leftarrow U\setminus U_f(t_\opt)$.
\ELSE
\STATE  Let $f_1,\ldots,f_k$ be $k$ disjointly supported functions such that $\supp(f_i)\subseteq \supp(f)$ and $$\phi_H(f)\leq O(k)\frac{\cR_H(f)}{\sqrt{\max_{1\leq i\leq k}\cR_H(f_i)}},$$ 
as defined in \autoref{t:mainfunc}.
\STATE Find  a threshold set $S=U_{f_i}(t)$ for $1\leq i\leq k$, and $t>0$ such that 
$$w(E(S,U\setminus S)) \leq w(E(S,V\setminus U)).$$
\vspace{-.5cm}
\STATE $U\leftarrow U\setminus S$.
\ENDIF
\ENDWHILE
\RETURN $\overline{U}$.
\end{algorithmic}
\caption{A Spectral Algorithm for Balanced Separator}
\label{alg:minbisection}
\end{algorithm}

Let $U$ be the set of  vertices remained after a number of steps of the induction, where initially $U=V$.
We will maintain the invariant that $\phi_G(\overline{U}) \leq O(k \eps/\lambda_k)$. 
Suppose $\vol(U)> \frac45 \vol(V)$. 
Let $H=(U,E(U))$ be the induced subgraph of $G$ on $U$, 
and $0=\lambda'_1\leq \lambda'_2\leq \ldots$ be the eigenvalues of $\cL_H$.
First, observe that $\lambda'_2 = O(\eps)$ as the following lemma shows.

\begin{lemma}
For any set $U\subseteq V$ with $\vol(U)\geq \frac45 \vol(V)$, let $H(U,E(U))$ be the induced subgraph of $G$ on $U$. Then the second smallest eigenvalue $\lambda'_2$ of $\cL_H$ is at most $10\eps$.
\end{lemma}
\begin{proof}
Let $(T,\overline{T})$ be the optimum bisection, and let $T':=U\cap T$. 
Since $\vol(U)\geq \frac45 \vol(V)$, and $\vol(T)= \vol(V)/2$, we have 
$$\vol_H(T')\geq \vol_G(T)-2\vol_G(\overline{U}) \geq \vol(V)/2 - 2\vol(V)/5 = \vol(V)/10 = \vol(T)/5.$$
Furthermore, since $E(T',U\setminus T') \subseteq E(T,\overline{T})$, 
we have 
$$\phi_H(T') = \frac{w(E(T',U \setminus T'))}{\vol_H(T')} \leq \frac{w(E(T,\overline{T}))}{\vol_G(T)/5} \leq 5\phi(T)= 5\eps.$$
Therefore, by the easy direction of Cheeger's inequality \eqref{eq:Cheegerineq}, we have $\lambda'_2\leq 10\eps$.
\end{proof}

To prove \autoref{thm:minbisection}, it is sufficient to find a set $S\subseteq U$ with $\vol_H(S) \leq \frac12 \vol_H(U)$ and conductance $\phi_H(S)\leq O(k \lambda'_2/\lambda_k) = O(k \eps/\l_k)$, because
$$\phi_G(\overline{U} \cup S) \leq 
\frac{w(E_G(\overline{U}, U)) + w(E_H(S, \overline{S}))}{\vol_G(\overline{U}) + \vol_H(S)} \leq \max(\phi_G(\overline{U}),\phi_H(S))) \leq O(k \eps/\l_k),$$ 
and so we can recurse until $\frac{1}{5} \vol(V) \leq \vol(\overline{U} \cup S) \leq \frac{4}{5} \vol(V)$.
Let $f\in \ell^2(U,w)$ be a non-negative
function such that $\vol_H(\supp(f))\leq \frac12 \vol_H(U)$ and $\cR_H(f)\leq \lambda'_2$, 
as defined in \autoref{c:two}.  
If $\phi_H(f)\leq O(k\lambda'_2/\lambda_k)$, then we are done. 
Otherwise, we will
find a set $S$ such that $\vol_H(S)\leq \frac{1}{2} \vol_H(U)$ and $w(E(S,U\setminus S))\leq w(E(S,\overline{U}))$. This implies
that we can simply remove $S$ from $U$ without increasing the expansion of the union of the removed vertices, because $\phi_G(S\cup \overline{U}) \leq \phi_G(\overline{U})$ as the numerator (total weight of the cut edges) does not increase while the denominator (volume of the set) can only increase. 

It remains to find a set $S$ with either of the above properties.  
We can assume that $\phi_H(f)\nleq O(k)\cR_H(f)$ as otherwise we are done.
Then, by \autoref{t:mainfunc}, there are $k$ disjointly supported functions $f_1,\ldots,f_k\in \ell^2(U,w)$  such that $\supp(f_i)\subseteq \supp(f)$ and
$$\phi_H(f)\leq O(k)\frac{\lambda'_2 }{\sqrt{\max \cR_H(f_i)}}.$$
We extend $f_i \in \ell^2(U,w)$ to $f_i \in \ell^2(V,w)$ by defining $f_i(v)=0$ for $v \in V-U$.
We will prove that either $\phi_H(f)\leq O(k\lambda'_2/\lambda_k)$, or there is a threshold set $S=V_{f_i}(t)$ for some $1\leq i\leq k$ and $t>0$ such that $ w(E(S,U\setminus S)) \leq w(E(S,\overline{U}))$.
As $f_1, \ldots, f_k$ can be computed in polynomial time, this will complete the proof of \autoref{thm:minbisection}.

Suppose that for every $f_i$ and any threshold set $S=V_{f_i}(t)$ we have $w(E(S,\overline{U})) \leq w(E(S,U\setminus S))$. Then, by \autoref{lem:rayleighenlargement} that we will prove below, $\cR_H(f_i) \geq \Omega(\cR_G^2(f_i))$ for every $1\leq i\leq k$. 
This implies that
$$ \phi_H(f) \leq O(k)\frac{\lambda'_2}{\sqrt{\max_{1\leq i\leq k} \cR_H(f_i)}} \leq 
O(k) \frac{\lambda'_2}{\sqrt{\max_{1\leq i\leq k} \cR_G^2(f_i)}} \leq O(k)\frac{\lambda'_2}{\lambda_k},$$
where the last inequality follows by \autoref{c:orthogonal} and the fact that $f_1,\ldots,f_k$ are
disjointly supported. 

\begin{lemma}
\label{lem:rayleighenlargement}
For any set $U\subseteq V$, let $H(U,E(U))$ be the induced subgraph of $G$ on $U$, and  $f\in \ell^2(V,w)$ be a non-negative function such that $f(v)=0$ for any $v\notin V-U$. 
Suppose that for any threshold set $V_f(t)$, we have
$$w(E(V_f(t),\overline{U})) \leq w(E(V_f(t), U\setminus V_f(t))),$$
then
$$ \sqrt{8\cR_H(f)} \geq \cR_G(f).$$
\end{lemma}
\begin{proof}
Since both sides of the inequality are homogeneous in $f$, we may assume that $\max_v f(v)\leq 1$. Furthermore, we can assume that $\sum_v w(v)f^2(v)=1$ (this is achievable since we assumed that $w(v)\geq 1$ for all $v\in V$). Observe that, since $w_H(v)\leq w_G(v)$ for all $v\in U$,
\begin{equation}
\label{eq:Hrayleighdenom2}
\sum_{v\in U} w_H(v) f^2(v) \leq  \sum_{v\in U} w_G(v) f^2(v) = \sum_v w_G(v) f^2(v) =1.
\end{equation}
Let $0< t\leq 1$ be chosen uniformly at random. Then, by linearity of expectation,
\begin{eqnarray}
\E{w(E(V_f(\sqrt t), U\setminus V_f(\sqrt t)))} &=& \sum_{(u,v)\in E(U)}w(u,v) |f^2(u)-f^2(v)| \nonumber \\
&=& \sum_{(u,v)\in E(U)} w(u,v) |f(u)-f(v)| |f(u)+f(v)|\nonumber\\
&\leq& \sqrt{\sum_{(u,v)\in E(U)} w(u,v) |f(u)-f(v)|^2} \sqrt{\sum_{(u,v)\in E(U)} w(u,v) (f(u)+f(v))^2}\nonumber\\
&\leq & \sqrt{2\cR_H(f)}. \label{eq:insideHedges}
\end{eqnarray}
where the first equality uses the fact that $f(v)\leq 1$ for all $v\in V$, and the last inequality follows by \eqref{eq:Hrayleighdenom2}. 
On the other hand, since $w(E(V_f(t),\overline{U})) \leq w(E(V_f(t), U\setminus V_f(t)))$ for any $t$,
\begin{eqnarray}
\E{w(E(V_f(\sqrt t), U\setminus V_f(\sqrt t)))} &\geq& \frac12\E{w(E(V_f(\sqrt t), \overline{V_f(\sqrt t)}))} \nonumber\\
&=& \frac12 \sum_{u\sim v} w(u,v) |f^2(u)-f^2(v)|\nonumber\\
&\geq & \frac12 \sum_{u\sim v} w(u,v) |f(u)-f(v)|^2 = \frac12 \cR_G(f).\label{eq:HGedges}
\end{eqnarray}
where the last inequality follows by the fact that $f(v)\geq 0$ for all $v\in V$, and the last equality  follows by the normalization $\sum_v w(v)f^2(v)=1.$ 
Putting together \eqref{eq:insideHedges} and \eqref{eq:HGedges} proves the lemma.
\end{proof}

\def\mlambda{\alpha}
\def\ml{\alpha}
\def\X{L}
\def\Y{R}
\subsection{Maximum Cut} \label{s:maxcut}

%
In this subsection we show that our techniques can be extended to the maximum cut 
problem providing a new spectral algorithm with its approximation ratio in terms of higher eigenvalues of the graph.

Let $\cM_G:= I + D^{-1/2} A D^{-1/2}$.
Observe that $\cM$ is a positive semi-definite matrix,
and an eigenvector with eigenvalue $\ml$ of $\cM$ is an eigenvector with eigenvalue $2-\ml$ of $\cL$.
We use $0\leq \mlambda_1 \leq \cdots \leq \mlambda_n\leq 2$ to denote its eigenvalues. 
In this section we analyze a polynomial time approximation algorithm for the Maximum Cut problem using the higher eigenvalues of $\cM$.
We restate \autoref{thm:maxcut-intro} as follows.

\begin{theorem}
\label{thm:maxcut}
There is a polynomial time algorithm that on input graph $G$ finds a cut $(S,\overline{S})$ such that if the optimal solution cuts at least $1-\eps$ fraction of the edges, 
then $(S,\overline{S})$ cuts at least 
$$1-O(k)\log(\frac{\mlambda_k}{k\eps})\frac{\eps}{\mlambda_k}$$ 
fraction of edges. 
\end{theorem}

The structure of this algorithm is similar to the structure of the algorithm for the balanced separator problem, with the following modifications.
First, we use the bipartiteness ratio of an induced cut defined in~\cite{trevisan09} in place of the conductance of a cut.
Then, similar to the first proof of \autoref{t:main}, we show that the spectral algorithm in~\cite{trevisan09} returns an induced cut with bipartiteness ratio $O(k \ml_1 / \sqrt{\ml_k})$.
Finally, we iteratively apply this improved analysis along with an additional step to obtain a cut with the performance guaranteed in \autoref{thm:maxcut}.

For an induced cut $(\X,\Y)$ such that $\X \cup \Y\neq \emptyset$, the {\em bipartiteness ratio} of $(\X,\Y)$ is defined as

\[ \beta(\X,\Y) := \frac{2w(E(\X)) + 2w(E(\Y)) + w(E(\X\cup \Y, \overline{\X\cup \Y}))}{\vol(\X\cup \Y)}. \]


The bipartiteness ratio $\beta(G)$ of $G$ is the minimum of $\beta(\X,\Y)$ over all induced cuts $(\X,\Y)$. 
For a function $f\in \ell^2(V,w)$ and a threshold $t\geq 0$, let $\X_f(t):=\{v: f(v)\leq -t\}$ and $\Y_f(t):=\{v: f(v)\geq t\}$ be a threshold cut of $f$. 
We let
$$ \beta(f):=\min_{t\geq 0} \beta(\X_f(t),\Y_f(t))$$
be the bipartiteness ratio of the best threshold cut of $f$, 
and let $(L_f(t_\opt),R_f(t_\opt))$ be the best threshold cut of $f$.
The following lemma is proved in \cite{trevisan09} and the proof is a simple extension of \autoref{lem:dirichletcheeger}.
\begin{lemma}[\cite{trevisan09}]
\label{lem:trevisanbipartiteness}
For every non-zero function $h\in \ell^2(V,w)$, 
$$\beta(h) \leq \frac {\sum_{u\sim v}w(u,v) |h(v) + h(u)| }{\sum_v w(v) |h(v)| }.$$
\end{lemma}

In this section we abuse the notation and write $\cR(f)$, the Rayleigh quotient of $f$, as
$$ \cR(f):=\frac{\sum_{u\sim v} w(u,v) |f(u)+f(v)|^2}{\sum_v w(v) f(v)^2}.$$
This is motivated by the fact that the eigenfunctions of $\cM$ are the optimizers of the above ratio. In particular, using the standard variational principles and \autoref{c:orthogonal}, 
\begin{eqnarray*}
\mlambda_k &=& 
 \min_{f_1, \ldots, f_k \in \ell^2(V,w)} \max_{f \neq 0}  \left\{ \vphantom{\bigoplus} \cR(f) : f \in \mathrm{span}\{f_1, \ldots, f_k\}\right\} \\
 &\leq & 2\min_{\substack{f_1,\ldots,f_k\in\ell^2(V,w) \\ \text{disjointly supported}}} \max_{1\leq i\leq k} \cR(f_i),
\end{eqnarray*}
where the first minimum is over sets of $k$ non-zero orthogonal functions in the Hilbert space $\ell^2(V,w)$, and the second minimum is over sets of $k$ disjointly supported functions in $\ell^2(V,w)$.

Trevisan~\cite{trevisan09} proved the following characterization of the bipartiteness ratio in terms of $\mlambda_1$.
\begin{theorem}[\cite{trevisan09}]
\label{thm:trevisanbipartiteness}
For any undirected graph $G$, 
\[ \frac {\mlambda_1}2 \leq \beta(G) \leq \sqrt{2 \mlambda_1} \]
\end{theorem}

We improve the right hand side of the above theorem and prove the following. 
\begin{theorem}
\label{thm:bipartiteness} 
For any function $f\in\ell^2(V,w)$ and any $1\leq k\leq n$,
\[ \beta(f) \leq 16\sqrt{2} k \cdot \frac{\cR(f)}{\sqrt{\mlambda_k}}. \]
Therefore, letting $\cR(f)= \mlambda_1$ implies $\beta(G)\leq O(k\mlambda_1/\sqrt{\mlambda_k}).$
\end{theorem}
The proof of the above theorem is an adaptation of the proof of \autoref{t:main}.
Let $f$ be the eigenfunction corresponding to $\mlambda_1$ with $\norm{f}_w=1$. 
The main difference is that here we can not assume $f$ is non-negative. In fact most of the edges of the graph
will have endpoints of different signs. 

The rest of this section is organized as follows. First we prove \autoref{thm:bipartiteness}
in \autoref{subsec:bipartiteness}. Then we prove \autoref{thm:maxcut} in \autoref{subsec:maxcut}.

\subsubsection{Improved Bounds on Bipartiteness Ratio}

\label{subsec:bipartiteness}
We say a function $g\in \ell^2(V,w)$ is a $2k+1$ step approximation of $f$,
if there exists thresholds $0=t_0\leq t_1\leq \ldots\leq t_{2k}$ such that for
any $v\in V$,
$$ g(v) = \psi_{-t_{2k}, -t_{2k-1},\ldots,-t_1,0,t_1,\ldots,t_{2k}}(f(v)).$$
In words, $g(v)$ is the value in the set $\{-t_{2k}, -t_{2k-1},\ldots,-t_1,0,t_1,\ldots,t_{2k}\}$
that is closest to $f(v)$. Note that here for every threshold $t$ we  include a symmetric threshold $-t$ in the step function. 
The proof of the next lemma is an adaptation of \autoref{l:jump}.

\begin{lemma} \label{lm:ellone}
For any non-zero function $f\in \ell^2(V,w)$ with $\norm{f}_w=1$, and any $2k+1$-step approximation of $f$, called $g$,
\[\beta(f) \leq  4k\cR(f) + 4\sqrt2k \norm{f-g}_w \sqrt{ \cR(f) }. 
\]
\end{lemma}
\begin{proof}
Similar to \autoref{l:jump}, we will construct a function $h\in \ell^2(V,w)$ such that
$$ \frac {\sum_{u \sim v} w(u,v)|h(u)+h(v)|}{\sum_{v\in V} w(v) |h(v)|} \leq  4k\cR(f) + 4\sqrt2k \norm{f-g}_w \sqrt{ \cR(f)}, $$
then the lemma follows from \autoref{lem:trevisanbipartiteness}.
Let $g$ be a $2k+1$ step approximation of $f$ with thresholds $0=t_0\leq t_1\leq \ldots\leq t_{2k}$. 
Let $\mu(x):= |x - \psi_{-t_{2k},\ldots,-t_{1},0,t_{1},\ldots,t_{2k}}(x)|$. 
We define $h$ as follows:
\[ h(v):= \int_{0}^{f(v)} \mu(x) dx. \]
Note that if $f(v) \leq 0$ then $h(v) := -\int_{f(v)}^{0} \mu(x)dx$.
First, by \autoref{cl:denom},

\begin{equation}
\label{eq:hlower}
 |h(v)| \geq   \frac{|f(v)|^2}{8k}.
 \end{equation}

It remains to prove that for every edge $(u,v)$,

\begin{equation}
\label{eq:upvupper}
|h(v) + h(u)| \leq \frac12 |f(v) + f(u)| \cdot ( |f(v) + f(u)| + |g(v)-f(v)| + |g(u)-f(u)|).
\end{equation}

If $f(u)$ and $f(v)$ have different signs, then using the fact that $\mu(x) = \mu(-x)$,   
\[|h(u) + h(v)| = \Big|\int_{0}^{f(u)} \mu(x) dx + \int_{0}^{f(v)} \mu(x) dx \Big| = |\int_{-f(v)}^{f(u)} \mu(x) dx| \leq |f(u)+f(v)| \cdot \max_{x \in [f(u),-f(v)]} \mu(x),\]
and thus \eqref{eq:upvupper} follows from the proof of \autoref{cl:numerator}
which shows that $\max_{x \in [f(u),-f(v)]} \mu(x) \leq \frac{1}{2} (|f(u)+f(v)|+|g(v)-f(v)|+|g(u)-f(u)|)$.
On the other hand, if $f(u)$ and $f(v)$ have the same sign, say that they are both positive, then since $|\mu(x)| \leq |x|$ for all $x$, we get

\begin{eqnarray*} |h(v) + h(u)| \leq   \int_{0}^{f(v)} x dx + \int_{0}^{f(u)} x dx \leq \frac12 |f(v) + f(u)|^2.
\end{eqnarray*}
Putting together \eqref{eq:hlower} and \eqref{eq:upvupper}, the lemma follows from a similar proof as in \autoref{l:jump}.
\end{proof}

\autoref{thm:bipartiteness} follows simply from the following lemma,
which is an adaptation of \autoref{l:lambda_k}.

\begin{lemma} 
\label{lem:ksteps}
For any non-zero function $f\in \ell^2(V,w)$ with $\norm{f}_w=1$, 
at least one of the following holds:
\begin{enumerate}[i)]
\item $\beta(f) \leq 8k\cR(f).$
\item There exist $k$ disjointly supported functions $f_1,\ldots,f_k$ such that for all $1\leq i\leq k$,
$$ \cR(f_i) \leq 256 k^2 \frac{\cR^2(f)}{\beta^2(f)}. $$
\end{enumerate}
\end{lemma}
\begin{proof}
Let $M:=\max_v |f(v)|$. We find $2k+1$ thresholds $0=t_0\leq t_1\leq \ldots\leq t_{2k}=M$,
and define $g$ to be a $2k+1$ step approximation of $f$ with respect to these thresholds.
Let $$C:=\frac{\beta^2(f)}{256k^3\cR(f)}.$$
We choose the thresholds inductively. Given
$t_0, t_1,\ldots,t_{i-1}$, we let $t_{i-1}\leq t_i \leq M$ be the smallest number such that
\begin{equation}
\label{eq:equalizenormcut} \sum_{v:-t_i \leq f(v) \leq -t_{i-1}} w(v) |f(v)-\psi_{-t_i,-t_{i-1}}(f(v))|^2 
+ \sum_{v:t_{i-1}\leq f(v)\leq t_i} w(v) |f(v) - \psi_{t_{i-1},t_i}(f(v))|^2 = C.
\end{equation}
Similar to the proof of \autoref{l:lambda_k}, the left hand side varies continuously with $t_i$, and it is non-decreasing. 
If we can satisfy \eqref{eq:equalizenormcut}  for some $t_{i-1}\leq t_i<M$,
then we let $t_i$ to be the smallest such number; otherwise we set $t_i=M$.

If $t_{2k}=M$ then we say the procedure succeeds. 
We show that if the procedure succeeds then (i) holds, 
and if it fails then (ii) holds.
First, if the procedure succeeds, then we can define $g$ to be the $2k+1$ step approximation of $f$ with respect to $t_0,\ldots,t_{2k}$, and by \eqref{eq:equalizenormcut} we get 
$$\norm{f-g}_w^2\leq 2kC = \frac{\beta^2(f)}{128k^2\cR(f)}.$$ 
By \autoref{lm:ellone}, this implies that
$$\beta(f) \leq 4k\cR(f) + \frac{\beta(f)}{2}, $$
and thus part (i) holds.

If the procedure does not succeed, then we will construct $k$ disjointly supported functions of Rayleigh quotients less than $1/kC$ and that would imply (ii). 
For each $1\leq i\leq 2k$, we let $f_i$ be the following function,
$$ 
f_i(v):=\begin{cases}
-|f(v)-\psi_{-t_i,-t_{i-1}}(f(v))| & \text{if } -t_i\leq f(v)\leq-t_{i-1} \\
|f(v) - \psi_{t_{i-1},t_i}(f(v))| & \text{if }  t_{i-1}\leq f(v)\leq t_i\\
0 & \text{otherwise}.
\end{cases}
$$
We will argue that at least $k$ of these functions satisfy $\cR(f_i) < 1/kC$.
By \eqref{eq:equalizenormcut}, we already know that the denominators of $\cR(f_i)$ are equal to $C$, so it remains to find an upper bound for the numerators. 
For each pair of vertices $u,v\in V$,
we will show that
\begin{equation}\label{eq:lipschitzsum} \sum_{i=1}^{2k} |f_i (u)+f_i(v)|^2 \leq |f(u) + f(v)|^2. \end{equation}
Note that $u,v$ are contained in the support of at most two of the functions. We distinguish three cases:

\begin{itemize}
\item $u$ and $v$ are in the support of the same function $f_i$. Then \eqref{eq:lipschitzsum} holds since each $f_i$ is 1-Lipschitz.
\item $u\in\supp(f_i)$ and $v\in\supp(f_j)$ for $i\neq j$, and $f(u),f(v)$  have the same sign. Then \eqref{eq:lipschitzsum} holds since
$$ |f_i(u)+f_i(v)|^2 + |f_j(u)+f_j(v)|^2 = |f_i(u)|^2 + |f_j(v)|^2 \leq |f(u)|^2 + |f(v)|^2\leq |f(u) + f(v)|^2. $$
\item $u\in\supp(f_i)$ and $v\in\supp(f_j)$ for $i\neq j$, and $f(u),f(v)$ have different signs. Then \eqref{eq:lipschitzsum} holds by \eqref{eq:lipschitzfi}.
\end{itemize}

Summing inequality \eqref{eq:lipschitzsum}, we have

\[ \sum_{i=1}^{2k} \cR(f_i) = \frac1C \sum_{i=1}^{2k} \  \sum_{(u,v) \in E } 
w(u,v)|f_i (u)+f_i(v)|^2 \leq \frac1C \sum_{(u,v) \in E} w(u,v)|f(u)+f(v)|^2 = 256 k^3\frac{\cR^2(f)}{\beta^2(f)}.\]
By an averaging argument, there are $k$ functions of Rayleigh quotients less than $256 k^2\cR^2(f)/\beta^2(f)$, and thus (ii) holds.
\end{proof}

\subsubsection{Improved Spectral Algorithm for Maximum Cut}
\label{subsec:maxcut}

In this section we prove \autoref{thm:maxcut}.
Our algorithm for max-cut is very similar to \autoref{alg:minbisection}.
We inductively remove an induced cut such that the union of removed vertices cuts a large fraction of the edges. 
The detailed algorithm is described in \autoref{alg:maxcut}.

\begin{algorithm}
\begin{algorithmic}
\STATE $U\leftarrow V$, $\X\leftarrow \emptyset$, $\Y\leftarrow \emptyset$.
\WHILE {$E(U)\neq \emptyset$}
\STATE Let $H=(U,E(U))$ be the induced subgraph of $G$ on $U$.
\STATE Let $f$ be the first eigenvector of $\cM$.
\IF {$\beta_H(f) \leq 192\sqrt{2} k \cR(f) / \mlambda_k$}
\STATE $(\X,\Y) \leftarrow (\X\cup \X_f(t_\opt), \Y\cup \Y_f(t_\opt))$.
\ELSE 
\STATE Let $f_1,\ldots,f_k$ be $k$ disjointly supported functions such that
$$ \beta_H(f) \leq 16 k \frac{\cR_H(f)}{\sqrt{\max_{1\leq i\leq k} \cR_H(f_i)}},$$
as defined in \autoref{lem:ksteps}.
\STATE Find a threshold cut $(\X',\Y')=(\X_{f_i}(t),\Y_{f_i}(t))$ for $1\leq i\leq k$ such that
$$\min(\gamma(\X\cup\X',\Y\cup\Y'), \gamma(\X\cup\Y', \Y\cup\X')) \leq \gamma(\X,\Y).$$
\vspace{-0.5cm}
\STATE Remove $\X',\Y'$ from $U$, and let $(\X,\Y)$ be one of $(\X\cup\X',\Y\cup\Y')$ or $(\X\cup\Y',\Y\cup\X')$ with minimum uncutness.
\ENDIF
\ENDWHILE
\RETURN $(\X,\Y)$.
\end{algorithmic}
\caption{A Spectral Algorithm for Maximum Cut}
\label{alg:maxcut}
\end{algorithm}

For technical reasons we define a parameter called uncutness to measure the total weight of cut edges throughout the algorithm. 
For an induced cut $(\X,\Y)$, the {\em uncutness} of $(\X,\Y)$ is defined as
$$ \gamma(\X,\Y):=w(E(\X)) + w(E(\Y)) + w(E(\X\cup\Y,\overline{\X\cup\Y})).$$
In words, it is the total weight of the edges adjacent to $\X$ and $\Y$ that are not  $E(\X,\Y)$.
Note that the coefficient of edges inside $\X$ and $\Y$ is one (instead of two as in the definition of bipartiteness ratio).

Throughout the algorithm we maintain an induced cut $(\X,\Y)$.
To extend this induced cut, we either find an induced cut $(\X', \Y')$ in the remaining graph with bipartiteness ratio $O(k\cR_H(f)/\ml_k)$, 
or an induced cut $(L',R')$ such that $\gamma(\X \cup \X',\Y \cup \Y') \leq \gamma(\X,\Y)$.
We will show later that this would imply \autoref{thm:maxcut}.

Let $(\X,\Y)$ be the cut extracted after a number of steps of the induction, and let $U=V\setminus (\X\cup \Y)$ be the set of the remaining vertices. Let
$H=(U,E(U))$ be the induced subgraph of $G$ on $U$ and $0=\mlambda'_1 \leq \mlambda'_2\leq \ldots$ be the eigenvalues of $\cM_H$. Furthermore, assume that
$w(E(U))= \rho \cdot w(E(V))$ where $0 < \rho \leq 1$. Since the optimal solution cuts at least $1-\eps$ (weighted) fraction of edges of $G$, it must cut at least $1-\eps/\rho$ (weighted) fraction of the edges of $H$. Therefore,
by \autoref{thm:trevisanbipartiteness}, 
$$\mlambda'_1 \leq 2\eps/\rho.$$

First, if $\beta_H(f)\leq 192\sqrt{2} k \cR_H(f)/\mlambda_k$, then we find the best threshold cut $(\X',\Y')=(\X_f(t_\opt),\Y_f(t_\opt))$ of $f$, 
and update $(\X,\Y)$ to 
$(\X\cup \X',\Y\cup \Y')$, and remove $\X' \cup \Y'$ from $H$,
  and recurse.
Otherwise, by \autoref{lem:ksteps}, there are $k$ disjointly supported functions $f_1,\ldots,f_k$ such that for all $1\leq i\leq k$, 
$$ \beta_H(f) \leq 16k \frac{\cR_H(f)}{\sqrt{\max_{1\leq i\leq k} \cR_H(f_i)}}.$$
Next, we show that we can find a threshold cut $(\X',\Y')$ of one of these functions such that 
$$\min(\gamma(\X\cup \X', \Y\cup\Y'),\gamma(\X\cup\Y',\Y\cup\X')) \leq \gamma(\X,\Y).$$
In words, we can merge $(\X',\Y')$ with the set of removed vertices such that 
the uncutness of the extended induced cut, say $(\X \cup \X', \Y \cup \Y')$, does not increase.
To prove this claim we use \autoref{lem:maxcutenlargement} which will be proved below. 
By \autoref{lem:maxcutenlargement}, if we can not find such a threshold cut for each of the functions $f_1,\ldots,f_k$, then we must have
$$\cR_H(f_i) \geq \frac{1}{72}\cR_G^2(f_i)$$
for all $1\leq i\leq k$.
Henceforth,
$$ \beta_H(f) \leq 16k \frac{\cR_H(f)}{\sqrt{\max_{1\leq i\leq k} \cR_H(f_i)}} \leq
96\sqrt{2} k \frac{\mlambda'_1}{\sqrt{\max_{1\leq i\leq k} \cR_G^2(f_i)}} \leq 192\sqrt{2} k \frac{\cR_H(f)}{\mlambda_k}.$$
where the last inequality follows by \autoref{c:orthogonal}, and the assumption that $f_1,\ldots,f_k$ are disjointly supported. This is a contradiction.
Therefore, we can always either find a threshold cut $(\X',\Y')$ of $f$ such that 
$$\beta(\X',\Y')\leq 192\sqrt{2} k \frac{\cR_H(f)}{\mlambda_k} \leq 600k\frac{\eps}{\rho \mlambda_k},$$ 
or we can remove an induced cut from $H$ while making sure that the uncutness of the induced cut does not increase. We keep doing this until $E(U)=\emptyset$.

It remains to calculate the ratio of the edges cut by the final solution of the algorithm. 
Let $\rho_j \cdot w(E)$ be the fraction of edges in $H$ before the $j$-th iteration of the for loop for all $j\geq 1$, in particular $\rho_1=1$.

Suppose the first case holds, i.e. we choose a threshold cut of $f$ with small bipartitness ratio.
Then we cut at least $(1-600 k \eps /\rho_j\mlambda_k)$ fraction of the edges  removed from $H$ in the $j$-th iteration.
Since the weight of the edges in the $j+1$ iteration is  $\rho_{j+1} w(E)$, we can lower-bound the weight of the cut edges by
$$ (\rho_j w(E) - \rho_{j+1}w(E))(1-600 k \frac{\eps}{\rho_j \mlambda_k}) \geq w(E)\int_{\rho_{j+1}}^{\rho_j} (1-600k\frac{\eps}{r \mlambda_k}) dr.$$

Suppose the second case holds, i.e. we choose a threshold cut of one of $f_1,\ldots,f_k$. 
Then, since the uncutness does not increase, the weight of the newly cut edges in the $j$-th iteration is at least as large as the total weight of the edges removed from $H$ in the $j$-th iteration. 
In other words, the total weight of the edges cut in the $j$-th iteration is at least $\rho_j w(E) - \rho_{j+1}w(E)$ in this case.

Putting these together, the fraction of edges cut by \autoref{alg:maxcut} is at least
$$ \int_{600 k \eps/\mlambda_k}^1 \Big(1-600k \frac{\eps}{r\mlambda_k}\Big) dr = 1-\frac{600k \eps}{\mlambda_k} \Big(1+\ln\Big(\frac{\mlambda_k}{600k\eps}\Big)\Big).$$
This completes the proof of \autoref{thm:maxcut}

\begin{lemma}
\label{lem:maxcutenlargement}
For any set $U\subseteq V$, let $H(U,E(U))$ be the induced subgraph of $G$ on $U$, and $f\in \ell^2(V,w)$ be a non-zero function such that $f(v)=0$ for any $v \notin U$. Also let $(\X,\Y)$ be a partitioning of $\overline{U}$. If for any threshold cut $(\X_f(t),\Y_f(t))$, 
\begin{equation}
\label{eq:bipartitenessinvariant} 
\min\big(\gamma_G(\X\cup \X_f(t),\Y\cup \Y_f(t)), \gamma_G(\X\cup \Y_f(t), \Y\cup\X_f(t)) \big)> \gamma_G(\X,\Y),
\end{equation}
then
$$ \sqrt{72\cR_H(f)} \geq \cR_G(f).$$
\end{lemma}
\begin{proof} 
First, observe that if 
\[\frac12w(E(\X_f(t)\cup \Y_f(t), \overline{U})) 
 > w(E(\X_f(t))) + w(E(\Y_f(t))) + w(E(\X_f(t)\cup \Y_f(t), U\setminus (\X_f(t)\cup \Y_f(t)))),\]
then \eqref{eq:bipartitenessinvariant} does not hold for that $t$.
Therefore, if \eqref{eq:bipartitenessinvariant} holds for any threshold cut $(\X_f(t),\Y_f(t))$ of $f$, then we have (the weaker condition) that
\begin{equation}
\label{eq:maxcutweaker} 
\frac12w(E(\X_f(t)\cup \Y_f(t), \overline{U})) 
 \leq  2w(E(\X_f(t))) + 2w(E(\Y_f(t))) + w(E(\X_f(t)\cup \Y_f(t), U\setminus (\X_f(t)\cup \Y_f(t)))).
\end{equation}
Henceforth, we prove the lemma by showing that  $\sqrt{72 \cR_H(f)} \geq \cR_G(f)$ holds whenever \eqref{eq:maxcutweaker} holds for any threshold cut of $f$.

Since both sides of \eqref{eq:maxcutweaker} are homogeneous in $f$, we may assume that $\max_v f(v)\leq 1$. Furthermore, we can assume that $\sum_{v\in V} w(v)f^2(v)=1$. Observe that, since $w_H(v)\leq w_G(v)$ for all $v\in U$,
\begin{equation}
\label{eq:Hrayleighdenom}
\sum_{v\in U} w_H(v) f^2(v) \leq  \sum_{v\in U} w_G(v) f^2(v) = \sum_v w_G(v) f^2(v) =1.
\end{equation}
Let $0< t\leq 1$ be chosen uniformly at random. 
For any vertex $v$, let $Z_v$ be the random variable where
$$ Z_v = \begin{cases}
1 & \text{if } f(v)\geq \sqrt{t}\\
-1 & \text{if } f(v) \leq -\sqrt{t}\\
0&\text{otherwise}.
\end{cases}
$$
\begin{claim}
\label{cl:maxcutedgecutting}
For any edge $\{u,v\}\in E$, 
$$ \frac12 |f(u)+f(v)|^2 \leq \E{|Z_u+Z_v|} \leq |f(u)+f(v)| (|f(u)|+|f(v)|).$$
\end{claim}
\begin{proof}
Without loss of generality assume that $|f(u)| \leq |f(v)|$.
We consider two cases.
\begin{itemize}
\item If $f(u)$ and $f(v)$ have different signs, then $|Z_u + Z_v| = 1$ when $|f(u)|^2< t< |f(v)|^2$.
Therefore,  
$$\E{|Z_u+Z_v|} = |f(v)|^2 - |f(u)|^2 = |f(u)+f(v)| (|f(u)|+|f(v)|),$$
and the claim holds.
\item If $f(u)$ and $f(v)$ have the same sign, then 
$$|Z_u+Z_v|=\begin{cases}
2& \text{if } t<|f(u)|^2,\\
 1 & \text{if } |f(u)|^2 \leq t < |f(v)|^2,\\
 0 & \text{if } |f(v)|^2\leq t.
\end{cases}$$
Therefore, 
$$ \frac12 (f(u)+f(v))^2 \leq \E{|Z_u+Z_v|} = f(u)^2 + f(v)^2 \leq (f(u)+f(v))^2. \hfill\qedhere$$
\end{itemize}
\end{proof}

The rest of the proof is very similar to that in \autoref{lem:rayleighenlargement}.
\begin{eqnarray}
&  & \mathbb{E} [ 2w(E(\X(\sqt))) + 2w(E(\Y(\sqt))) + w(E(\X(\sqt)\cup \Y(\sqt), U\setminus (\X(\sqt)\cup \Y(\sqt)))) ] \nonumber\\
&= & \sum_{(u,v)\in E(U)}w(u,v) \E{|Z_u+Z_v|} \nonumber \\
&\leq & \sum_{(u,v)\in E(U)} w(u,v) |f(u)+f(v)| (|f(u)|+|f(v)|)\nonumber\\
&\leq& \sqrt{\sum_{(u,v)\in E(U)} w(u,v) |f(u)+f(v)|^2} \sqrt{\sum_{(u,v)\in E(U)} w(u,v) (|f(u)|+|f(v))^2}\nonumber\\
&\leq & \sqrt{2\cR_H(f)}, \label{eq:maxcutinsideHedges}
\end{eqnarray}
where the first inequality follows by \autoref{cl:maxcutedgecutting}, and the last inequality follows by \eqref{eq:Hrayleighdenom}. 
On the other hand, 
by \eqref{eq:maxcutweaker},
\begin{eqnarray}
& & \mathbb{E}[2w(E(\X(\sqt))) + 2w(E(\Y(\sqt))) + w(E(\X(\sqt)\cup \Y(\sqt), U\setminus (\X(\sqt)\cup \Y(\sqt))))]\nonumber\\
& \geq & \frac13 \mathbb{E}[2w(E(\X(\sqt))) + 2w(E(\Y(\sqt))) + w(E(\X(\sqt)\cup \Y(\sqt), V\setminus (\X(\sqt)\cup(\Y\sqt))))]\nonumber \\
&= & \frac13 \sum_{u\sim v} w(u,v) \E{|Z_u+Z_v|} \nonumber\\
&\geq & \frac16 \sum_{u\sim v} w(u,v) |f(u)+f(v)|^2 = \frac16 \cR_G(f),\label{eq:maxcutHGedges}
\end{eqnarray}
where the second inequality follows from \autoref{cl:maxcutedgecutting}, and the last equality  follows from the normalization $\sum_v w(v)f^2(v)=1.$ 
Putting together \eqref{eq:maxcutinsideHedges} and \eqref{eq:maxcutHGedges} proves the lemma.
\end{proof}

\subsection{Manifold Setting} \label{s:manifold}

The eigenvalues of a closed Riemannian manifold can be approximated
by the eigenvalues of the Laplacian of the graph of a $\eps$-net
in $M$~\cite{fujiwara}.
Hence, \autoref{t:main} implies a generalized Cheeger's inequality for closed Riemannian manifolds.

\begin{theorem} \label{t:manifold}
Let $M$ be a $d$-dimensional closed Riemannian manifold. Let $\lambda_{k}(M)$
be the $k^{th}$ eigenvalue of Laplacian of $M$ and $\phi(M)$ be
the Cheeger isoperimetric constant of $M$. Then
\[
\phi(M)\leq Ck\frac{\lambda_{2}(M)}{\sqrt{\lambda_{k}(M)}}
\]
where $C$ depends on $d$ only.
\end{theorem}

\subsection{Planted and Semi-Random Instances}
 \label{s:semirandom}

As discussed in the introduction, spectral techniques can be used to recover the hidden bisection when $p-q \geq \Omega(\sqrt{p \log |V|/|V|})$ in the planted random model~\cite{boppana,mcsherry}, and for other hidden partition problems~\cite{alon-krivelevich-sudakov,mcsherry}.
Some semi-random models have been proposed and the results in planted random models can be generalized using semidefinite programming relaxations~\cite{feige-kilian,makarychev+}:
Feige and Kilian~\cite{feige-kilian} considered the model where a planted instance is generated and an adversary is allowed to delete arbitrary edges between the parts and add arbitrary edges within the parts, and they proved that an SDP-based algorithm can recover the hidden partition when $p-q \geq \Omega(\sqrt{p \log |V|/|V|})$.
Makarychev, Makarychev and Vijayaraghaven~\cite{makarychev+} considered a more flexible model where the induced subgraph of each part is arbitrary, and proved that an SDP-based algorithm would find a balanced cut with good quality.
These results show that SDP-based algorithms are more powerful than spectral techniques for semi-random instances.

For graph bisection, we note that there will be a gap between $\l_2$ and $\l_3$ in the instances in the planted random model 
when $p-q$ is large enough.
\autoref{t:main} shows that the spectral partitioning algorithm performs better in instances just satisfying this ``pseudorandom'' property, 
although the bounds are much weaker when applied to random planted instances.
For example, our result implies that the spectral partitioning algorithm performs better in the following ``deterministic'' planted instances where there are two arbitrary bounded degree expanders of size $|V|/2$ with an arbitrary bounded degree sparse cut between them.  

\begin{corollary} \label{c:planted}
Let $G = (V, E)$ be an unweighted graph such that 
$V=A\cup B$, where $\vol(A) = \vol(B)$ and $\phi(A) = \phi(B) = \phi$.
Let $G_A$ and $G_B$ be the induced subgraphs of $G$ on $A$ and $B$, 
and $\varphi = \min(\phi(G_A), \phi(G_B))$.
Suppose that the minimum degree in $G_A$ and $G_B$ is at least $d_1$,
and the maximum degree of the bipartite subgraph $G' = (A \cup B, E(A, B))$ is at most $d_2$.
Then the spectral partitioning algorithm applied to $G$ returns a set of conductance 
\[O \left( \frac{\phi}{\varphi} \frac{d_1 + d_2}{d_1} \right).\]
\end{corollary}
\begin{proof}
We call $\{S_1,S_2,S_3\}$ a $3$-partition of $V$ if $S_1,S_2,S_3$ are disjoint and $S_1 \cup S_2 \cup S_3 = V$.
We will show that any $3$-partition of $V$ contains a set of large conductance.
This implies that $\phi_3(G)$ is large, and thus $\lambda_3$ is large by the higher-order Cheeger's inequality.
Then \autoref{t:main} will prove the corollary.

Given a $3$-partition, let $S$ be the set of smallest volume, then $\vol(S) \le \vol(V)/3$.
We show that 
\begin{equation}
\label{eq:phi3bound}
\phi_G(S) \geq \frac{\varphi d_1 }{ 2(d_1+d_2)}.
\end{equation}
Let $m = |E(S) \cap E(A, B)|$ be the number of induced edges in $S$ that cross $A$ and $B$.
Then $|S| \ge 2m/d_2$, since the total degree of $S$ in $G'$ is at least $2m$ but the maximum degree in $G'$ is at most $d_2$.
Observe that 
\begin{eqnarray*}
|E(S,\overline{S})| &\geq& |E(S \cap A, A - S \cap A)| + |E(S \cap B, B - S \cap B)| \\
&\geq& \frac12\phi_{G_A}(S \cap A) \cdot \vol(S \cap A) + \frac12\phi_{G_B}(S \cap B) \cdot \vol(S \cap B)\\
& \geq& \frac{\varphi}{2} (\vol(S)-2m),
\end{eqnarray*}
where the second inequality follows by the fact that $\vol(S) \le \frac23\vol(A) = \frac23\vol(B)$,
and the last inequality follows by $\phi(G_A)\geq \varphi$ and $\phi(G_B)\geq \varphi$.
Therefore,
\[
\phi(S) = \frac{|E(S, \overline{S})|}{\vol(S)}
\ge \frac{\varphi \cdot (\vol(S) - 2m)}{2\vol(S)}
= \frac{\varphi}2 - \frac{m \cdot \varphi}{\vol(S)}
\ge \frac{\varphi}2 - \frac{m \cdot \varphi}{2m + d_1|S|}
\ge \frac{\varphi}2 - \frac{\varphi}{2 + 2d_1/d_2}
= \frac{\varphi}2\frac{d_1}{d_1 + d_2},
\]
where the last inequality uses the fact that $|S| \geq 2m/d_2$.
This proves \eqref{eq:phi3bound}. Therefore, $\phi_3(G) \geq \varphi d_1/2(d_1+d_2)$.
But by the higher order Cheeger's inequality, 
$\phi_3(G) = O(\sqrt{\lambda_3})$.  
Therefore, \autoref{t:main} implies that the spectral partitioning algorithm returns a set of conductance
\[
O(\frac{\l_2}{\sqrt{\l_3}}) = O(\frac{\phi}{\varphi} \frac{d_1+d_2}{d_1}).
\]
\end{proof}

We note that the degree requirements on $d_1$ and $d_2$ are necessary.
Otherwise, the bipartite graph $G'$ may only contain a heavy edge (with weight $\phi \cdot \vol(V)/2$) connecting $u \in A$ and $v \in B$ where $d_{G_A}(u) = d_{G_B}(v) = 1$.
Then $A - \{u\}, B - \{v\}, \{u, v\}$ are all sparse cuts and $\lambda_3 \approx \lambda_2$, and \autoref{t:main} would not apply.

\autoref{c:planted} implies that the spectral partitioning algorithm is a constant factor approximation algorithm for planted random instances.
Let $G = (A \cup B, E)$ be a graph such that $|A| = |B| = |V|/2$,
where each induced edge in $A$ and each induced edge in $B$ appears with probability $p$ and each edge crossing $A$ and $B$ appears with probability $q$.
Suppose $p > q > \Omega(\ln n/n)$, then with high probability $\vol(A) \approx \vol(B)$, $\phi(A) \approx \phi(B) \approx q/(p+q)$ and $\phi(G_A) \approx \phi(G_B) \approx \Theta(1)$.
Putting the parameters $\phi \approx q/(p+q)$, $\varphi \approx \Theta(1)$, $d_1 \approx pn$, $d_2 \approx qn$, \autoref{c:planted} implies that the spectral partitioning algorithm returns a set of conductance $O(q/p)$.

\subsection{Stable Instances} 
\label{s:stable}

Several clustering problems are shown to be easier on stable instances~\cite{balcan-blum-gupta,awasthi-blum-sheffet,daniely-linial-saks}, but there are no known results on the stable sparsest cut problem.
As discussed earlier, the algebraic condition that $\l_2$ is small and $\l_3$ is large is of similar flavour to the condition that there is a stable sparse cut, but they do not imply each other.
On one hand, using the definition of stability in the introduction, one can construct an instance with an $\Omega(n)$-stable sparse cut but the gap between $\l_2$ and $\l_3$ is $O(1/n^2)$:
Suppose the vertices are $\{1,...,2n\}$.  
There is an odd cycle, $1$-$3$-$5$-$7$-$9$-$...$-$(2n-1)$-$1$ where each edge is of weight $1$.  
There is an even cycle $2$-$4$-$6$-$8$-$10$-$...$-$2n$-$2$ where each edge is of weight $1$.  
There is an edge between $2i-1$ and $2i$ for each $1 \leq i \leq n$, 
where each edge is of weight $c/n^2$ for a constant $c$.  
Then the optimal cut is the set of odd vertices with conductance $1/n^2$,
and this is an $\Omega(n)$ stable sparse cut.
But the second eigenvector and the third eigenvector will be the same as in the cycle example (if $c$ is a large enough constant), where the vertices are in the order $1,2,3,...,2n$ and the Rayleigh quotients are of order $1/n^2$.  

On the other hand, it is not hard to see that an instance with a large gap between $\l_2$ and $\l_3$ is not necessarily $1$-stable, because there could be multiple optimal sparse cuts.
A more relaxed stability condition is that any near-optimal sparse cut is ``close'' to any optimal solution.
More precisely, we say a cut $(S,\overline{S})$ is $\epsilon$-closed to an optimal cut $(T,\overline{T})$ if the fraction of their symmetric difference $\delta = \vol(S \Delta T)/\vol(V)$ satisfies $\delta < \epsilon$ or $\delta > 1 - \epsilon$.
We call an instance to the sparsest cut problem $(c, \epsilon)$-stable if any $c$-approximation solution is $\epsilon$-close to any optimal solution.
It is possible to show that if $\l_2$ is small and $\l_3$ is large then the instance is stable under this more relaxed notion. 

\begin{corollary}
Any instance to the sparsest cut problem is $(c, \Theta(c\lambda_2/\lambda_3^{3/2}))$-stable for any $c \geq 1$.
\end{corollary}

\begin{proof}
Let $(T,\overline{T})$ be an optimal cut with $\vol(T) \leq \vol(V)/2$ and $\phi = \phi(T)$.
Suppose the instance is not $(c, \epsilon)$-stable.
Then there exists a cut $(S,\overline{S})$ of conductance at most $c \phi$ and $\vol(S) \leq \vol(V)/2$ such that $\vol(S \Delta T)/\vol(V) \in [\epsilon, 1 - \epsilon]$.
Let $S_1$ be $S - T$ or $T - S$, whichever of larger volume.
Let $S_2$ be $S \cap T$ or $V - S - T$, whichever of larger volume.
Then, by our assumption, we have $\vol(S_i) \ge \epsilon \cdot \vol(V)/2$ for $i = 1, 2$.
Also, for $i=1,2$,
\[w(E(S_i, \overline{S_i})) \le w(E(S,\overline{S})) + w(E(T,\overline{T})) \leq \phi \cdot \vol(T) + c\phi \cdot \vol(S) \le (1 + c)\phi \cdot \vol(V)/2.\]
Therefore $\phi(S_i) \le (1 + c)\phi / \eps$.
Finally, observe that $S_3 := V - S_1 - S_2$ is one of these four sets: $T$, $S$, $\overline{T}$, $\overline{S}$.
This implies that $\phi(S_3) \le c\phi/\epsilon$.
Thus,
\[\lambda_3 \le 2\max_i \phi(S_i) = O(\frac{c\phi}{\epsilon}) = O(\frac{c\lambda_2}{\epsilon\sqrt{\lambda_3}}),\]
where the last inequality follows from \autoref{t:main}.
Therefore $\epsilon = O(c\lambda_2/\lambda_3^{3/2})$.
\end{proof}

There is also another interpretation of our result through numerical stability.
By the Davis-Kahan theorem from matrix perturbation theory (see~\cite{von-luxburg}),
when there is a large gap between $\l_2$ and $\l_3$,
then the second eigenvector is stable under perturbations of the edge weights of the graph.
More generally, when there is a large gap between $\l_k$ and $\l_{k+1}$, then the top $k$-dimensional eigenspace is stable under perturbations of the edges weights of the graph.
Our result shows that spectral partitioning performs better when the top eigenspace is stable.
Some similar results are known in other applications of spectral techniques~\cite{azar+,von-luxburg2}.

\bibliographystyle{plain}

\begin{thebibliography}{60}

\bibitem[Alo86]{alon}
N. Alon.
{\em Eigenvalues and expanders}.
Combinatorica, 6, 83--96, 1986.

\bibitem[AM85]{alon-milman}
N. Alon, V. Milman.
{\em $\l_1$, isoperimetric inequalities for graphs, and superconcentrators}.
Journal of Combinatorial Theory, Series B, 38(1), 73--88, 1985.

\bibitem[AKS98]{alon-krivelevich-sudakov}
N. Alon, M. Krivelevich, B. Sudakov.
{\em Finding a large hidden clique in a random graph}.
Random Structures and Algorithms 13(3-4), 457--466, 1998.

\bibitem[ABS10]{arora-barak-steurer}
S. Arora, B. Barak, D. Steurer.
{\em Subexponential algorithms for unique games and related problems}.
In Proceedings of the 51st Annual IEEE Symposium on Foundations of Computer Science (FOCS), 563--572, 2010.

\bibitem[ARV04]{arora-rao-vazirani}
S. Arora, S. Rao, U. Vazirani.
{\em Expander flows, geometric embeddings and graph partitioning}.
In Proceedings of the 36th Annual ACM Symposium on Theory of Computing (STOC), 222--231, 2004.

\bibitem[ABS10]{awasthi-blum-sheffet}
P. Awasthi, A. Blum, O. Sheffet.
{\em Stability yields a PTAS for k-median and k-means clustering}.
In Proceedings of the 51st Annual IEEE Symposium on Foundations of Computer Science (FOCS), 309--318, 2010.

\bibitem[AFKMS01]{azar+}
Y. Azar, A. Fiat, A.R. Karlin, F. McSherry, J. Saia.
{\em Spectral analysis of data}.
In Proceedings of the 33rd Annual ACM Symposium on Theory of Computing (STOC), 619--626, 2001.

\bibitem[BBG09]{balcan-blum-gupta}
M.-F. Balcan, A. Blum, A. Gupta.
{\em Approximate clustering without the approximation}.
In Proceedings of the 20th Annual ACM-SIAM Symposium on Discrete Algorithms (SODA), 1068--1077, 2009

\bibitem[BL10]{bilu-linial}
Y. Bilu, N. Linial.
{\em Are stable instances easy}?
In Proceedings of Innovations in Computer Science (ICS), 332--341, 2010.

\bibitem[BDLS12]{bilu-daniely-linial-saks}
Y. Bilu, A. Daniely, N. Linial, M. Saks.
{\em On the practically interesting instances of MAXCUT}.
In arXiv:1205.4893, 2012.

\bibitem[BLR10]{biswal-lee-rao}
P. Biswal, J.R. Lee, S. Rao.
{\em Eigenvalue bounds, spectral partitioning, and metrical deformations via flows}.
Journal of the ACM 57(3), 2010.

\bibitem[Bop87]{boppana}
R. Boppana.
{\em Eigenvalues and graph bisection: An average-case analysis}.
In Proceedings of the 28th Annual Symposium on Foundations of Computer Science (FOCS), 280--285, 1987.

\bibitem[Che70]{cheeger}
J. Cheeger.
{\em A lower bound for the smallest eigenvalue of the Laplacian}.
Problems in Analysis, Princeton University Press, 195--199, 1970.

\bibitem[Chu96]{Chung96}
F. R. K. Chung. 
{\em  Laplacians of graphs and Cheeger's inequalities. }
Combinatorics, Paul Erd\"os is eighty, Vol. 2 (Keszthely, 1993), volume 2 of Bolyai Soc. Math. Stud., pages 157--172. J\'anos Bolyai Math. Soc., Budapest, 1996.

\bibitem[Chu97]{Chung97}
Fan R. K. Chung. 
{\em Spectral graph theory.}
volume 92 of CBMS Regional Conference Series in Mathematics. Published for the Conference Board of the Mathematical Sciences, Washington, DC, 1997. 9

\bibitem[DLS12]{daniely-linial-saks}
A. Daniely, N. Linial, M. Saks.
{\em Clustering is difficult only when it does not matter}.
In arXiv:1205.4891, 2012.

\bibitem[FK01]{feige-kilian}
U. Feige, J. Kilian.
{\em Heristics for semirandom graph problems}.
J. Comput. Syst. Sci. 63, 639--673, 2001.

\bibitem[FK02]{feige-krauthgamer}
U. Feige, R. Krauthgamer.
{\em A polylogarithmic approximation of the minimum bisection}.
SIAM Journal on Computing 31(3), 1090--1118, 2002.

\bibitem[Fuj95]{fujiwara}
K. Fujiwara.
{\em Eigenvalues of Laplacians on a closed Riemannian manifold and its nets}.
Proceedings of the American Mathematical Society 123(8), 1995.

\bibitem[GW95]{goemans-williamson}
M.X. Goemans, D.P. Williamson.
{\em Improved approximation algorithms for maximum cut and satisfiability problems using semidefinite programming}.
Journal of the ACM 42(6), 1115--1145, 1995.

\bibitem[GM98]{guattery-miller}
S. Guattery, G.L. Miller.
{\em On the quality of spectral separators}.
SIAM J. Matrix Anal. Appl. 19(3), 701--719, 1998.


\bibitem[GS12]{guruswami-sinop2}
V. Guruswami, A.K. Sinop.
{\em Faster SDP hierarchy solvers for local rounding algorithms}.
In Proceedings of the 53rd IEEE Symposium on Foundations of Computer Science (FOCS), 2012.

\bibitem[HLW06]{hoory-linial-wigderson}
S. Horry, N. Linial, A. Wigderson.
{\em Expander graphs and their applications}.
Bulletin of the American Mathematical Society 43(4), 439--561, 2006.

\bibitem[JSV04]{jerrum-sinclair-vigoda}
M. Jerrum, A. Sinclair, and E. Vigoda. 
{\em A polynomial-time approximation algorithm for the permanent of a
matrix with nonnegative entries.}
 Journal of the ACM, 51(4):671Ð697,  2004.

\bibitem[JM85]{jimbo-maruoka}
S. Jimbo and A. Maruoka. 
{\em Expanders obtained from affine transformations. }
In Proceedings of the 17th Annual Symposiumon Theory of Computing (STOC),  88Ð97,  1985.

\bibitem[KVV04]{kannan-vempala-vetta}
R. Kannan, S. Vempala, A. Vetta.
{\em On cluterings: good, bad, and spectral}.
Journal of the ACM 51, 497--515, 2004.

\bibitem[Kel06]{kelner}
J. Kelner.
{\em Spectral partitioning, eigenvalue bounds, and circle packings for graphs of bounded genus}.
SIAM Journal on Computing 35(4), 882--902, 2006.

\bibitem[Lee12]{lee12}
J. R. Lee. 
{\em Gabber-galil analysis of margulis expanders. }\\
\protect\url{http://tcsmath.wordpress.com/2012/04/18/ gabber-galil-analysis-of-margulis-expanders/}, 2012.

\bibitem[KLPT11]{kelner-lee-price-teng}
J. Kelner, J.R. Lee, G. Price, S.-H. Teng.
{\em Metric uniformization and spectral bounds for graphs}.
Geom. Funct. Anal., 21(5), 1117--1143, 2011.

\bibitem[LLM10]{leskovec-lang-mahoney}
J. Leskovec, K.J. Lang, M.W. Mahoney.
{\em Empirical comparison of algorithms for network community detection}.
In Proceedings of the 19th International Conference on World Wide Web (WWW), 631--640, 2010.

\bibitem[LOT12]{lee-gharan-trevisan}
J.R. Lee, S. Oveis Gharan, L. Trevisan.
{\em Multi-way spectral partitioning and higher-order Cheeger inequalities}.
In Proceedings of the 44th Annual Symposium on Theory of Computing (STOC), 1117--1130, 2012.

\bibitem[LR99]{leighton-rao}
F.T. Leighton, S. Rao.
{\em Multicommodity max-flow min-cut theorem and their use in designing approximation algorithms}.
Journal of the ACM 46(6), 787--832, 1999.


\bibitem[LRTV12]{louis+2}
A. Louis, P. Raghavendra, P. Tetali, S. Vempala.
{\em Many sparse cuts via higher eigenvalues}.
In Proceedings of the 44th Annual ACM Symposium on Theory of Computing (STOC), 1131--1140, 2012.

\bibitem[Lux07]{von-luxburg}
U. von Luxburg.
{\em A tutorial on spectral clustering}.
Statistics and Computing 17(4), 395--416, 2007.

\bibitem[Lux10]{von-luxburg2}
U. von Luxburg.
{\em Clustering stability: an overview}.
Foundations and Trends in Machine Learning 2(3), 235--274, 2010.

\bibitem[MMV12]{makarychev+}
K. Makarychev, Y. Makarychev, A. Vijayaraghavan.
{\em Approximation algorithms for semi-random graph partitioning problems}.
In Proceedings of the 44th Annual ACM Symposium on Theory of Computing (STOC), 367--384, 2012.


\bibitem[McS01]{mcsherry}
F. McSherry.
{\em Spectral partitioning of random graphs}.
In Proceedings of the 42nd IEEE Symposium on Foundations of Computer Science (FOCS), 529--537, 2001.

\bibitem[NJW01]{ng-jordan-weiss}
A. Ng, M. Jordan, Y. Weiss.
{\em On spectral clustering: Analysis and an algorithm}.
Advances in Neural Information Processing Systems 14, 849--856, 2001.

\bibitem[OW12]{odonnell-witmer}
R. O'Donnell, D. Witmer.
{\em Improved small-set expansion from higher eigenvalues}.
CoRR, abs/1204.4688, 2012.

\bibitem[OT12]{gharan-trevisan}
S. {Oveis Gharan}, L. Trevisan.
{\em Approximating the expansion profile and almost optimal local graph clustering}.
In Proceedings of the 53rd Annual IEEE Symposium on Foundations of Computer Science (FOCS), 2012.

\bibitem[Rac08]{racke}
Harald R\"acke.
{\em Optimal hierarchical decompositions for congestion minimization in networks}.
In Proceedings of the 40th Annual ACM Symposium on Theory of Computing (STOC), 255--264, 2008.

\bibitem[SM00]{shi-malik}
J. Shi, J. Malik.
{\em Normalized cuts and image segmentation}.
IEEE Pattern Anal. Mach. Intell., 22(8), 888--905, 2000.

\bibitem[Shm97]{shmoys}
D.B. Shmoys.
{\em Approximation algorithms for cut problems and their applications to divide-and-conquer}.
In Approximation Algorithms for NP-hard Problems, (D.S. Hochbaum, ed.) PWS, 192--235, 1997.

\bibitem[SJ89]{sinclair-jerrum}
A. Sinclair and M. Jerrum. 
{\em Approximate counting, uniform generation and rapidly mixing markov chains. }
Inf. Comput., 82(1):93Ð133, July 1989.

\bibitem[ST07]{spielman-teng}
D.A. Spielman, S.-H. Teng.
{\em Spectral partitioning works: Planar graphs and finite element mashes}.
Linear Algebra and its Applicatiosn 421(2-3), 284--305, 2007.

\bibitem[Ste10]{steurer}
D. Steurer.
{\em On the complexity of unique games and graph expansion}.
Ph.D. thesis, Princeton University, 2010.

\bibitem[Tan12]{tanaka}
M. Tanaka.
{\em Higher eigenvalues and partitions of graphs}.
In arXiv:1112.3434, 2012.

\bibitem[Tre09]{trevisan09}
L. Trevisan.
{\em Max cut and the smallest eigenvalue}.
In Proceedings of the 41st Annual ACM Symposium on Theory of Computing (STOC),
263-272, 2009.

\bibitem[TM06]{tolliver-miller}
D.A. Tolliver, G.L. Miller.
{\em Graph partitioning by spectral rounding: Applications in image segmentation and clustering}.
In Proceedings of the 2006 IEEE Computer Society Conference on Computer Vision and Pattern Recognition (CVPR), 1053--1060.


\end{thebibliography}


\renewcommand{\a}{\alpha}

\begin{appendix}


\section{A New Proof of Cheeger's Inequality} \label{s:cheeger}

We will use \autoref{l:drop} to derive Cheeger's inequality with a weaker constant. 
The proof is a simplified version of our second proof of \autoref{t:main}.
By \autoref{c:two}, we assume that we are given a non-negative function  $f\in\ell^2(V,w)$ with $\cR(f) \leq \l_2$ and $\vol(\supp(f)) \leq \vol(V)/2$ and $\norm{f}_w=1$.
Fix $\a \in (0, 1)$.
Let $I_i = [\a^{i}, \a^{i+1}]$.
By \autoref{l:drop},
\[
\cE(I_i) \geq \frac{\phi^2(f) \cdot \vol^2(\a^{i}) \cdot \len^2(I_i)}{\phi(f) \cdot \vol(\a^i) + \vol(I_i)} 
\geq \frac{\phi^2(f) \cdot \vol^2(\a^{i}) \cdot \len^2(I_i)}{\vol(\a^{i}) + \vol(I_i)} 
= \frac{\phi^2(f) \cdot \vol^2(\a^{i}) \cdot \a^{2i}(1-\a)^2}{\vol(\a^{i+1})}.
\]
Summing over all intervals, we have
\begin{eqnarray*}
\cE_f  \geq  \sum_{i} \cE(I_i) \geq 
\sum_{i} \frac{\phi^2(f) \cdot \vol^2(\a^{i}) \cdot \a^{2i}(1-\a)^2}{\vol(\a^{i+1})}
&=& \phi^2(f) \cdot (1-\a)^2 \sum_{i \geq t} \frac{\vol^2(\a^{i}) \cdot \a^{4i}}{\vol(\a^{i+1}) \cdot \a^{2i}}\\
& \geq & \phi^2(f) \cdot (1-\a)^2 \frac{\big( \sum_{i} \vol(\a^{i}) \cdot \a^{2i}\big)^2 }{\sum_{i} \vol(\a^{i+1}) \cdot \a^{2i}}\\
&=& \phi^2(f) \cdot (1-\a)^2 \a^2 \sum_{i} \vol(\a^{i}) \cdot \a^{2i},
\end{eqnarray*}
where the third inequality follows from (\ref{eq:cauchy}). 
Changing the order of the summation,
\[
\sum_{i} \vol(\a^{i}) \a^{2i}
= \sum_{j} \vol(I_j) \sum_{i \geq j+1} \a^{2i}
= \frac{\a^2}{1-\a^2} \sum_{j} \vol(I_j) \a^{2j} 
\geq \frac{\a^2}{1-\a^2},
\]
where the last inequality holds by the assumption that $\norm{f}_w^2=1$.
Therefore,
\[\cE_f \geq \phi^2(f) \cdot (1-\a)^2 \a^2 \frac{\a^2}{1-\a^2}
= \phi^2(f) \frac{1-\a}{1+\a} \a^4.\]
Setting $\a = (\sqrt{17}-1)/4$, we get $\phi(f) < 4.68\sqrt{\l_2}$.


\section{A Different Proof of \autoref{t:main}}
\label{sec:jumpsecond}
In this section we give a different proof of \autoref{t:main}.
In particular, given a function $g$ that is a $2k+1$ step approximation of $f$, we
lower bound $\cR(f)=\cE_f$ using \autoref{l:drop}. 
This gives \autoref{l:jumpsec} which can be seen as a weaker version of \autoref{l:jump}. 
\begin{corollary} \label{l:step}
If $\cE_f\leq \lambda_k/(C^2k^2)$ for some constant $C$, then for any function $g\in\ell^2(V,w)$ satisfying \eqref{eq:fgclose},
$$ \norm{g}_w^2 \geq \left(1-\frac{4}{Ck}\right)^2. $$
\end{corollary}
\begin{proof}
The statement follows from a simple application of the triangle inequality:
\[\norm{g}_w^2 \geq (\norm{f}_w-\norm{f-g}_w)^2 \geq \left(1  - \sqrt{\frac{12\cE_f}{\lambda_k}}\right)^2 \geq \left(1-\frac{4}{Ck}\right)^2.\]
where the second inequality follows by \eqref{eq:fgclose}.
\end{proof}

\begin{proposition}
\label{l:jumpsec}
For any $2k+1$-step approximation of $f$, called $g$, 
$$ \cE_f  \geq \min\left\{\frac{\phi(f)\norm{g}_w^2}{32k}, \frac{\phi^2(f)\norm{g}_w^4}{2048k^2\norm{f-g}_w^2}\right\}.$$
\end{proposition}
\begin{proof}
Assume that $\range(g)=\{t_0,t_1,\ldots,t_{2k}\}$ such that $0=t_0\leq t_1\leq \ldots\leq t_{2k}$.
For each $1\leq i\leq 2k$, we let $I_i$ be the middle part of the interval $[t_{i-1},t_i]$, i.e.,
$$I_i:=\left[\frac{3t_{i-1}+t_i}{4},\frac{t_{i-1}+t_i}{2}\right].$$
Let $m_i:=(t_{i-1}+t_i)/2$ be the midpoint of $I_i$, and let $m_{2k+1}:=\infty$. Since the intervals are disjoint,  by \autoref{f:additivity} we can write
\begin{align}
\cE_f \geq \sum_{i=1}^{2k} \cE_f(I_i) \geq
\sum_{i=1}^{2k} \frac{\phi^2(f) \cdot \vol^2(m_i) \cdot \len^2(I_i)}{\phi(f) \cdot \vol(m_i)+\vol(I_i)}& = 
\frac{1}{16}\sum_{i=1}^{2k} \frac{\phi^2(f) \cdot \vol^2(m_i) \cdot (t_i-t_{i-1})^4}{\phi(f) \cdot \vol(m_i) \cdot (t_i-t_{i-1})^2+\vol(I_i) \cdot (t_i-t_{i-1})^2} \nonumber\\
&\geq  \frac{1}{16}\frac{\phi^2(f)\left(\sum_{i=1}^{2k} \vol(m_i)(t_i-t_{i-1})^2\right)^2}{\phi(f)\sum_{i=1}^{2k}\vol(m_i)(t_i-t_{i-1})^2 + \sum_{i=1}^{2k} \vol(I_i)(t_i-t_{i-1})^2},\label{eq:energyrelg}
\end{align}
where the second inequality follows by applying \autoref{l:drop} to each interval $I_i$,
and the third inequality follows from \eqref{eq:cauchy}.
Now to prove the proposition we simply use the following two claims.
\begin{claim}
$$ \sum_{i=1}^{2k} \vol(I_i) (t_i-t_{i-1})^2 \leq 16\norm{f-g}_w^2.$$
\end{claim}
\begin{proof}
Since $g$ is a $2k+1$ approximation of $f$, for any vertex $v$ such that $f(v)\in I_i$,
$$ |f(v) - g(v)| \geq \frac{t_i-t_{i-1}}{4}. $$
Therefore,
$$\norm{f-g}_w^2=\sum_{v} w(v)|f(v)-g(v)|^2 \geq \sum_{i=1}^{2k} \sum_{~v:f(v)\in I_i} w(v) |f(v)-g(v)|^2 \geq \frac{1}{16}\sum_{i=1}^{2k} \vol(I_i)(t_i-t_{i-1})^2. 
$$
\end{proof}

\begin{claim}
$$ \sum_{i=1}^{2k} \vol(t_i) (t_i-t_{i-1})^2 \geq \frac{\norm{g}_w^2}{2k}.$$
\end{claim}
\begin{proof}
The claim follows simply from changing the order of summations: 
\begin{eqnarray*}
\sum_{i=1}^{2k} \vol(m_i)(t_{i}-t_{i-1})^2 
= \sum_{i=1}^{2k} (t_{i}-t_{i-1})^2 \sum_{j=i}^{2k} (\vol(m_{j})-\vol(m_{j+1}))
& = & \sum_{i=1}^{2k} (\vol(m_{i})-\vol(m_{i+1})) \sum_{j=1}^{i} (t_{j}-t_{j-1})^2\\
&\geq& \sum_{i=1}^{2k} (\vol(m_{i}) - \vol(m_{i+1})) \frac{t_{i}^2}{2k} = \frac{\norm{g}_w^2}{2k}.
\end{eqnarray*}
where the first inequality follows from the Cauchy-Schwarz inequality, and the last equality follows
by the fact that for all vertices $v$ we have $g(v)=t_i$ when $m_i < f(v) \leq m_{i+1}$. 
\end{proof}
By \eqref{eq:energyrelg} and the above claims, we have
$$ \cE_f \geq \frac{\phi^2(f) \left(\sum_{i=1}^{2k} \vol(m_i)(t_i-t_{i-1})^2\right)^2}{16\phi(f)\sum_{i=1}^{2k} \vol(m_i)(t_i-t_{i-1})^2 + 256\norm{f-g}_w^2} \geq \min\left\{\frac{\phi(f)\norm{g}_w^2}{64k}, \frac{\phi^2(f)\norm{g}_w^4}{2048k^2\norm{f-g}_w^2}\right\} 
$$
\end{proof}

\begin{proofof}{\autoref{t:main}}
Let $g$ be as defined in \autoref{l:lambda_k}. If $\l_2 \geq \frac{\l_k}{256k^2}$, then by Cheeger's inequality,
\[\phi(f) \leq \sqrt{2\l_2} \leq \frac{32 k\l_2}{\sqrt{\l_k}},\]
and we are done. Otherwise, by \autoref{l:step}, we have $\norm{g}_w^2 \geq 1/2$.
Therefore, by \autoref{l:jump}, we have
\begin{eqnarray*}
 \cE_f   \geq \min\left\{\frac{\phi(f)\norm{g}_w^2}{32k}, \frac{\phi^2(f)\norm{g}_w^4}{2048k^2\norm{f-g}_w^2}\right\} \geq \frac{\phi^2(f)}{2^{13} k^2\norm{f-g}^2_w} \geq \frac{\lambda_k \phi^2(f)}{10^5 k^2 \cE_f},
\end{eqnarray*}
where the last inequality follows by \autoref{l:lambda_k}. Now the theorem
follows from the fact that $\cE_f=\cR(f)\leq \lambda_2$.
\end{proofof}


\end{appendix}

\end{document}